\documentclass[a4paper]{article}
\usepackage[american]{babel}
\usepackage{amsmath,amsfonts,amssymb,amsthm,bbm,mathrsfs}
\usepackage{graphicx}
\usepackage{enumerate,paralist}
\usepackage[T1]{fontenc}
\usepackage{pifont}

\newcommand{\abs}[1]{\left\vert #1\right\vert}
\newcommand{\B}{\mathcal{B}}
\newcommand{\bvarphi}{\boldsymbol{\varphi}}
\newcommand{\bphi}{\boldsymbol{\phi}}
\newcommand{\C}{\mathcal{C}}
\newcommand{\deltac}{\delta^c_{1j}}
\newcommand{\ds}[1]{\displaystyle{#1}}
\newcommand{\f}{\mathbf{f}}
\newcommand{\F}{\mathbf{F}}
\newcommand{\scrF}{\mathscr{F}}
\newcommand{\g}{\mathbf{g}}
\newcommand{\G}{\mathbf{G}}
\newcommand{\scrG}{\mathscr{G}}
\newcommand{\ind}{\mathbbm{1}}
\newcommand{\kk}{\mathsf{k}}
\newcommand{\scrL}{\mathscr{L}}
\newcommand{\Lip}{\operatorname{Lip}}
\newcommand{\n}{\mathsf{n}}
\newcommand{\norm}[1]{\Vert #1\Vert}
\newcommand{\R}{\mathbb{R}}
\newcommand{\Rn}{\R^n}

\renewcommand{\u}{\mathbf{u}}
\newcommand{\U}{\mathbf{U}}
\newcommand{\unit}[1]{\textup{#1}}

\theoremstyle{plain}\newtheorem{theorem}{Theorem}[section]
\theoremstyle{remark}\newtheorem{remark}[theorem]{Remark}
\theoremstyle{plain}\newtheorem{corollary}[theorem]{Corollary}
\theoremstyle{definition}\newtheorem{definition}[theorem]{Definition}

\usepackage{framed}
\newenvironment{tightcenter}{\setlength\topsep{5pt}\setlength\parskip{-5pt}\begin{center}}{\end{center}}
\newenvironment*{trafficrule}[2][\unskip]{\begin{framed}\begin{tightcenter}\textsc{\textbf{#2 rule} \textnormal{#1}}\end{tightcenter}\it}{\end{framed}}

\graphicspath{{./}{figure/}}

\title{A fully-discrete-state kinetic theory approach to traffic flow on road networks}
\author{Luisa Fermo\thanks{This author was partially supported by INdAM-GNCS (National Group for Scientific Computing of the National Institute for Advanced Mathematics, Italy).} \\
		{\small\it Department of Mathematics and Computer Science} \\
		{\small\it University of Cagliari} \\
		{\small\it Viale Merello 92, 09123 Cagliari, Italy} \\[5mm]
		Andrea Tosin \\
		{\small\it Istituto per le Applicazioni del Calcolo ``M. Picone''} \\
		{\small\it Consiglio Nazionale delle Ricerche} \\
		{\small\it Via dei Taurini 19, 00185 Rome, Italy}
	   }
\date{}

\begin{document}

\maketitle

\begin{abstract}
This paper presents a new approach to the modeling of vehicular traffic flows on road networks based on kinetic equations. While in the literature the problem has been extensively studied by means of macroscopic hydrodynamic models, to date there are still not, to the authors' knowledge, contributions tackling it from a genuine statistical mechanics point of view. Probably one of the reasons is the higher technical complexity of kinetic traffic models, further increased in case of several interconnected roads. Here such difficulties of the theory are overcome by taking advantage of a discrete structure of the space of microscopic states of the vehicles, which is also significant in view of including the intrinsic microscopic granularity of the system in the mesoscopic representation.

\medskip


\medskip

\end{abstract}

\section{Introduction}
Models of flows on networks are a consolidated and nonetheless continuously in ferment field of mathematical research from both theoretical and applied points of view, as the survey~\cite{bressan2014EMSSMS} demonstrates. In case of vehicular traffic, they bring single-road models to the next level, yielding a mathematical description able to cope with realistic scenarios involving many interconnected roads. This has relevant implications in terms of fruitful feedbacks between mathematical modeling and engineering applications toward the common goal of traffic forecast, see e.g.,~\cite{cristiani2010CAIM}.

A nowadays well-established theory of traffic flow on networks, whose milestone is the book~\cite{garavello2006BOOK}, uses partial differential equations, specifically conservation laws, at the macroscopic scale. On each edge, viz. road, of the network traffic is described with a hydrodynamic approach using average quantities: the car density $\rho$, i.e., the number of cars per kilometer, the mean speed $u$, and the flux $q=\rho u$. Road connections, viz. junctions, are treated as special boundary conditions for the conservation laws to be solved on the edges. In particular, transmission conditions among roads meeting at a junction have to be handled in such a way that the mass of cars flowing across the junction is conserved and that congestions (queues) possibly propagate backward, from outgoing to incoming roads, against the stream of vehicles. On the whole, the mathematical theory is quite complete and solid. Well-posedness of the problem as well as controllability and optimization of networked flows have been successfully addressed, see again~\cite{bressan2014EMSSMS,garavello2006BOOK} and references therein.

Macroscopic models give a synthetic view of the spatiotemporal evolution of traffic, focusing directly on group dynamics rather than on those of single vehicles. This makes them appealing in the context of (large) networks, where the main interest is normally in average traffic distributions. Smaller details of single cars would be indeed most of the times useless and, even worst, they would likely force one to deal with quite high numbers of unknown parameters. On the other hand, finer characterizations of the complexity of vehicle dynamics, for instance one-to-one interactions among cars responsible for speed variations, which pertain to a lower microscopic scale, are lost in this kind of description. An alternative to macroscopic models, closer to the microscopic states of the vehicles but still profiting from an aggregate description of traffic as opposed to an agent-based one, is a statistical approach by kinetic equations.

Kinetic modeling of vehicular traffic was initiated in~\cite{prigogine1961PROC}, subsequently refined in~\cite{paveri1975TR}, and then systematically developed in, among others, the series of papers~\cite{gunther2002MCM,gunther2003SIAP,klar1997JSP,wegener1996TTSP}. Readers are also referred to the surveys~\cite{bellomo2011SIREV,piccoli2009ENCYCLOPEDIA} for additional bibliographical details. These works were primarily concerned with either a Boltzmann or an Enskog-type treatment of vehicular flow inspired by analogies with classical gas dynamics. More recently new contributions have appeared~\cite{bellouquid2012M3AS,coscia2007IJNM,delitala2007M3AS,fermo2013SIAP,fermo2014DCDSS,tosin2009AML}, which use the \emph{generalized kinetic and stochastic game theory of active particles}~\cite{bellomo2013M3AS} together with discrete microscopic states of the vehicles. This way they successfully capture both the implicit stochasticity of human behaviors, hence ultimately of microscopic car interactions, and the intrinsic microscopic granularity of the distribution of vehicles along a road, which indeed do not properly form a continuum even in congested situations. All cited models are however confined to the treatment of traffic on single roads. Concerning networks, it is worth reporting that in~\cite{herty2009CMAM} a hybrid approach is used, which couples a macroscopic description of traffic on every road with a kinetic computation of the mass flow across the junctions. Nevertheless, at present a theory of car flow on networks genuinely grounded on a kinetic basis is, to the authors' knowledge, still missing.

The purpose of this paper is to contribute to filling such a gap by taking advantage of the fully-discrete-state kinetic model for single roads introduced in~\cite{fermo2013SIAP}, whose hallmarks are now briefly recalled for the sake of completeness.

As usual, the main dependent variable of the kinetic representation is the one-particle distribution function, say $f^r$ if it refers to a generic road indexed by $r$, over the microscopic state of the vehicles. The latter is classically given by the pair $(x,\,v)$, $x$ denoting the spatial position and $v$ the speed. On one-way roads traffic is described as one-dimensional, hence $x$, $v$ are scalar variables. Moreover, since it is customary to use dimensionless variables, throughout the paper space and speed will be assumed to be nondimensional referred to characteristic values $\ell,\,V>0$, respectively. A discrete representation of the space of microscopic states is obtained by introducing:
\begin{itemize}
\item A pairwise disjoint partition of the road in, say, $m_r\geq 1$ road cells denoted by $I^r_i$, which are assumed to have unit length, $\abs{I^r_i}=1$, for all $i,\,r$. Recalling the aforesaid nondimensionalization of the space, this implies that the characteristic length $\ell$ is the dimensional size of the cells, thereby defining the measure of the spatial granularity of the kinetic representation. Due to this argument, a conceivable reference value is $\ell=O(5~\unit{m})$, i.e., the order of magnitude of the average length of a compact car. \\
If road $r$ is identified with the interval $(0,\,m_r)\subset\R$ (i.e., its dimensionless length equals the number of cells of the partition), then the partition $\{I^r_i\}_{i=1}^{m_r}$ is, on the whole, such that:
$$ I^r_{i_1}\cap I^r_{i_2}=\emptyset\qquad\forall\,i_1\ne i_2, \quad\qquad \bigcup_{i=1}^{m_r}I^r_i=(0,\,m_r). $$
A road cell is the minimal space unit for vehicle localization. Because of the uncertainty due to the microscopic spatial granularity of the vehicle distribution, only the number of vehicles is known on average within each cell but not their pointwise location.
\item A speed lattice formed by, say, $n\geq 2$ distinct speed classes $v_j$ such that
$$ 0=v_1<\dots<v_j<v_{j+1}<\dots<v_n=1. $$
For instance, if the lattice is uniformly spaced then
\begin{equation}
	v_j=\frac{j-1}{n-1}, \quad j=1,\,\dots,\,n.
	\label{eq:vj_uniform}
\end{equation}
No speed is negative because vehicular traffic is one-directional. Moreover, with this choice of the dimensionless speed lattice the reference value $V$ turns out to be the maximum dimensional speed attainable by vehicles, or alternatively the maximum one allowed. Therefore it makes sense to assume $V=O(50~\unit{km/h})$, thinking of city roads.
\end{itemize}

Consequently the kinetic distribution function is given in the form:
$$ f^r(t,\,x,\,v)=\sum_{i=1}^{m_r}\sum_{j=1}^{n}f^r_{ij}(t)\ind_{I^r_i}(x)\delta_{v_j}(v), $$
where $\ind$ denotes the characteristic function ($\ind_{I^r_i}(x)=1$ if $x\in I^r_i$, $0$ otherwise) and $\delta$ the Dirac distribution. The coefficient $f^r_{ij}(t)$ is the density of vehicles which, at time $t>0$, travel in the $i$th cell of road $r$ with speed $v_j$. Coherently with the nondimensionalization introduced above, also the time variable has to be understood as dimensionless, in particular referred to the characteristic time $T=\ell/V=O(0.36~\unit{s})$. Likewise, the density $f^r_{ij}$ is referred to the maximum number of cars, say $N_\textup{max}=O(1)$, per unit length that can be accommodated in one road cell: $N_\textup{max}/\ell=O(200~\unit{car/km})$. Hence, in particular, $0\leq f^r_{ij}\leq 1$ for all $i,\,j,\,r$.

The macroscopic car density and flux can be computed, as usual in statistical mechanics, as zeroth and first order moments, respectively, of the distribution function with respect to $v$:
$$ \rho^r(t,\,x)=\int_0^1 f^r(t,\,x,\,dv)=\sum_{i=1}^{m_r}\left(\sum_{j=1}^{n}f^r_{ij}(t)\right)
	\ind_{I^r_i}(x)=\sum_{i=1}^{m_r}\rho^r_i(t)\ind_{I^r_i}(x), $$
where
\begin{equation}
	\rho^r_i:=\sum_{j=1}^{n}f^r_{ij}
	\label{eq:rhocell}
\end{equation}
is the car density in the cell $I^r_i$ referred to the characteristic value $N_\textup{max}/\ell$ (hence $0\leq\rho^r_i\leq 1$);
$$ q^r(t,\,x)=\int_0^1 vf^r(t,\,x,\,dv)=\sum_{i=1}^{m_r}\left(\sum_{j=1}^{n}v_jf^r_{ij}(t)\right)
	\ind_{I^r_i}(x)=\sum_{i=1}^{m_r}q^r_i(t)\ind_{I^r_i}(x), $$
where
\begin{equation}
	q^r_i:=\sum_{j=1}^{n}v_jf^r_{ij}
	\label{eq:qcell}
\end{equation}
is the flux crossing the cell $I^r_i$ referred to the characteristic value $N_\textup{max}V/\ell=O(10^4~\unit{car/h})$.

In order to substantiate a mathematical model, in~\cite{fermo2013SIAP} the following evolution equations for the $f^r_{ij}$'s are derived:
\begin{equation}
	\frac{df^r_{ij}}{dt}+v_j(\Phi(\rho^r_i,\,\rho^r_{i+1})f^r_{ij}-\Phi(\rho^r_{i-1},\,\rho^r_i)f^r_{i-1,j})
		=G_j[\f^r_i,\,\f^r_i]-f^r_{ij}L[\f^r],
	\label{eq:single.road.eq}
\end{equation}
$i=1,\,\dots,\,m_r$, $j=1,\,\dots,\,n$, where $\f^r_i:=\{f^r_{ij}\}_{j=1}^{n}$.

The left-hand side expresses the time variation of the number of cars in the cell $I^r_i$ and in the $j$th speed class due to the transport at their own speed $v_j$. The function $\Phi$ is called \emph{flux limiter} because it limits the free fluxes of vehicles traveling across adjacent cells, such as $I^r_{i-1}$, $I^r_i$, $I^r_{i+1}$, on the basis of the occupancy of the origin cell and the free room available in the destination cell. A prototypical form of $\Phi$ is:
\begin{equation}
	\Phi(\rho^r_i,\,\rho^r_{i+1})=\frac{\min\{\rho^r_i,\,1-\rho^r_{i+1}\}}{\rho^r_i},
		\qquad \Phi(0,\,\cdot)=1.
	\label{eq:fluxlim}
\end{equation}
In the sequel we will often use the shorthand notation $\Phi^r_{i,i+1}$ for $\Phi(\rho^r_i,\,\rho^r_{i+1})$.

The right-hand side expresses short-range conservative interactions among vehicles within the cell $I^r_i$, which can produce speed variations due to acceleration or braking. According to the already mentioned generalized kinetic theory of active particles such interactions are described as \emph{stochastic games} between pairs of players, viz. vehicles. The pre-interaction strategies of the interacting vehicles are their speeds at the moment that they meet, while the payoff of the game are their post-interaction speeds. Such payoffs are assumed to be known only in probability, whereby the intrinsic randomness of driver behaviors is taken into account in Eq.~\eqref{eq:single.road.eq}. The overall effect of the games is described by the \emph{gain} and \emph{loss operators}, respectively:
\begin{equation}
	G_j[\f^r_i,\,\f^r_i]=\eta_0\rho_i^r\sum_{h,\,k=1}^{n}A_{hk}^{j}[\rho^r_i]f^r_{ih}f^r_{ik}, \qquad
		L[\f^r_i]=\eta_0\rho^r_i\sum_{k=1}^{n}f^r_{ik}=\eta_0\left(\rho^r_i\right)^2,
	\label{eq:GL}
\end{equation}
which count statistically the number of interactions leading vehicles in the $i$th road cell to gain or lose, in the unit time, the speed $v_j$.

In Eq.~\eqref{eq:GL}, $\eta_0>0$ is a basic interaction frequency whereas the $A_{hk}^{j}[\rho^r_i]$'s are speed transition probabilities:
$$ A_{hk}^{j}[\rho^r_i]=\operatorname{Prob}(v_h\to v_j\vert v_h,\,v_k,\,\rho^r_i). $$
In practice, for fixed $h,\,k$, $\{A_{hk}^{j}[\rho^r_i]\}_{j=1}^{n}$ is the probability distribution of the payoff of an interaction between a vehicle traveling at speed $v_h$ and one traveling at speed $v_k$, further parameterized by the car density in the $i$th road cell, cf. Eq.~\eqref{eq:rhocell}. This way the influence of local traffic conditions on driver choices is taken into account. The whole set of probabilities $\{A_{hk}^{j}[\rho^r_i]\}_{h,\,k,\,j=1}^{n}$ forms the so-called \emph{table of games}. For its detailed expression readers are referred to~\cite{fermo2013SIAP}. A basic property of the table of games is that 
\begin{equation}
	\sum_{j=1}^{n}A_{hk}^{j}[\rho^r_i]=1, \qquad 
		\forall\,h,\,k=1,\,\dots,\,n, \quad \forall\,\rho^r_i\in [0,\,1],
	\label{eq:tog.sum.1}
\end{equation} 
whence using Eq.~\eqref{eq:GL} it results:
\begin{equation}
	\sum_{j=1}^{n}(G_j[\f^r_i,\,\f^r_i]-f^r_{ij}L[\f^r_i])=0
		\quad \forall\ i=1,\,\dots,\,m_r.
	\label{eq:conservativeness}
\end{equation}
This implies that, at a macroscopic level, Eq.~\eqref{eq:single.road.eq} entails the conservation of the car mass.

Model~\eqref{eq:single.road.eq} has some analogies with the Cell Transmission Model (CTM)~\cite{daganzo1994TRB,daganzo1995TRB}, especially as far as the space discretization in road cells is concerned. CTM was introduced in the engineering literature mainly as a computational simplification of macroscopic conservation-law-based models. Likewise, a by-product of model~\eqref{eq:single.road.eq}, particularly welcome in view of its application to networks, is that it reduces the technical complexity of standard kinetic models while being consistent with statistical mechanical theories of vehicular traffic.

As already anticipated, this paper studies how model~\eqref{eq:single.road.eq} can be profitably employed to describe traffic flows on a network of $R$ interconnected roads ($r=1,\,\dots,\,R$). The main issue of the theory is the modeling of transmission conditions at junctions, where the incoming roads merge and flows are redistributed among the outgoing roads. In Section~\ref{sect:modeling} the two most basic types of junctions are considered, which serve as building blocks for all other more complex cases: a junction with one incoming and two outgoing roads, which introduces the concept of flow redistribution, and the dual junction with two incoming and one outgoing road, which covers the case of road merging and introduces the concept of right-of-way. Then in Section~\ref{sect:numerics} selected exploratory examples show that representative networks, for instance traffic circles, can be studied relying on only these two basic junctions. Finally, in Section~\ref{sect:analysis} the picture is completed by a well-posedness theory of the initial/boundary-value problems across such junctions.

\section{Kinetic modeling of networks}
\label{sect:modeling}
\subsection{Basic concepts}
\label{sect:basic}
A road network can be represented as an oriented graph in the plane, whose edges and vertexes are the roads and the junctions, respectively, see Fig.~\ref{fig:gennet}. From now on, with a slight abuse of speech, we will often confuse the words edge, road and vertex, junction. The orientation of an edge indicates the direction in which vehicles flow along the corresponding road, assuming one-dimensional unidirectional traffic. Two-way roads are possibly represented by two distinct edges connecting the same vertexes but having opposite orientations. Actually not all vertexes of the graph correspond strictly to junctions in the network. If in a given vertex there are both incoming and outgoing edges then that vertex is a true junction (black bullets in Fig.~\ref{fig:gennet}). If, conversely, there are only incoming or outgoing edges then that vertex is actually a peripheral one, namely either an exit or an access point for the modeled portion of the network (white squares in Fig.~\ref{fig:gennet}). 
\begin{figure}[!t]
\includegraphics[width=\textwidth]{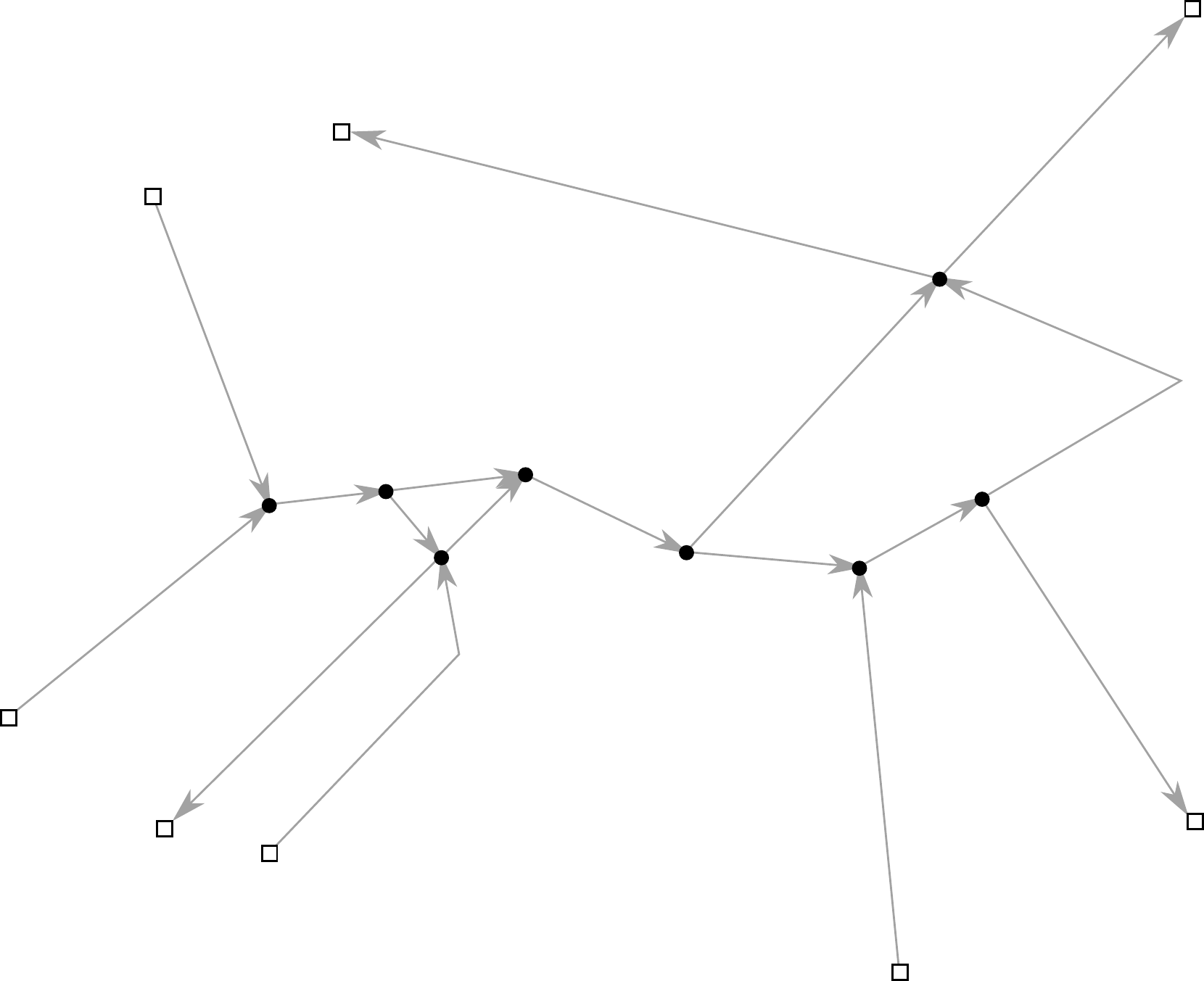}
\caption{Representation of a generic road network as an oriented graph}
\label{fig:gennet}
\end{figure}

The car flow on the $r$th edge is modeled by means of Eq.~\eqref{eq:single.road.eq} using a convenient number $m_r$ of space cells according to the actual length of the corresponding road. At the endpoints of the edge boundary conditions have to be prescribed. In fact, in the first cell ($i=1$) Eq.~\eqref{eq:single.road.eq} requires $\{f^r_{0j}\}_{j=1}^{n}$ (the flux limiter $\Phi^r_{0,1}$ is then univocally determined by Eq.~\eqref{eq:fluxlim}) while in the last cell ($i=m_r$) it requires $\Phi^r_{m_r,\,m_r+1}$. We incidentally notice that the ``first'' and ``last'' cells are identified on the basis of the orientation of the edge under consideration. Determining such boundary conditions in the proper way amounts to integrating the single-road traffic model~\eqref{eq:single.road.eq} with flow dynamics at junctions.

At peripheral vertexes only one of the aforesaid two types of boundary conditions is needed. More precisely, at an access point one assigns the distribution functions $\{f^r_{0j}\}_{j=1}^{n}$ to model the upstream flow of vehicles coming from a part of the network not explicitly modeled and entering road $r$. At an exit point one assigns instead the flux limiter $\Phi^r_{m_r,m_r+1}$ to model traffic conditions downstream, which may allow or prevent vehicles from leaving road $r$ at their free flux $v_jf^r_{m_rj}$. Conversely, at an actual junction the previous two types of boundary conditions apply simultaneously and take more properly the form of \emph{transmission conditions}. Outgoing roads require the distribution functions $\{f^r_{0j}\}_{j=1}^{n}$, which are determined by the outflow of vehicles from incoming roads. On the other hand, incoming roads require the flux limiter $\Phi^r_{m_r,m_r+1}$, which depends on the inflow of vehicles that outgoing roads can sustain based on their current traffic loads. Such a mutual exchange of information between incoming and outgoing roads is subject to the \emph{mass conservation constraint}: the total incoming flux must equal the total outgoing one, so that no cars come out or remain trapped in the junction.

It is worth stressing a relevant difference between the kinetic and the macroscopic frameworks, which has consequences on the way in which transmission conditions are managed. Since in kinetic equations the transport speed coincides with the microscopic speed of cars, which is nonnegative by definition, the characteristic lines keep always the same direction over time. Thus, once the orientation of the road has been fixed, inflow and outflow boundaries invariably coincide with the first and last endpoint of the road, respectively. Consequently, the relevant information about the structure of the network is fully topological and can be summarized in a $J\times R$ incidence matrix, say $\mathbb{I}$, where $J$ is the total number of junctions of the network. We recall for completeness that the generic entry $\mathbb{I}_{sr}$ of such a matrix is defined as:
$$ \mathbb{I}_{sr}=
	\begin{cases}
		-1 & \text{if road\ } r\ \text{enters junction\ } s\ \text{(incoming road)} \\
		1 & \text{if road\ } r\ \text{leaves junction\ } s\ \text{(outgoing road)} \\
		0 & \text{if road\ } r\ \text{and junction\ } s\ \text{are not incident}
	\end{cases}
$$
or possibly with the opposite convention of sign. Conversely, in macroscopic equations the characteristic speeds do not coincide, in general, with the speed of the vehicles. Hence inflow and outflow boundaries can appear, disappear, or swap over at different times depending on whether the characteristic speed is instantaneously positive or negative in the neighborhood of each endpoint of the road. Consequently, unlike the kinetic case, inflow vs. outflow transmission conditions at a junction cannot be univocally associated to the outgoing vs. incoming character of the roads hitting that junction, which entails additional technicalities to be handled.

\subsection{Basic types of junctions}
In this section we discuss how the concepts set forth above, together with a few other case-dependent rules, constitute the basis for modeling the two most fundamental types of junctions mentioned in the Introduction, namely:
\begin{itemize}
\item a junction with one incoming and two outgoing roads, that we will henceforth refer to briefly as \emph{1-2 junction};
\item a junction with two incoming and one outgoing road, that we will henceforth refer to briefly as \emph{2-1 junction}.
\end{itemize}
Notice incidentally that even quite complex networks, such as the one represented in Fig.~\ref{fig:gennet}, can in most cases be almost entirely reduced to a combination of only these two types of junctions.

Before tackling them, however, we preliminarily speculate on a fictitious junction formed by only one incoming and one outgoing road (\emph{1-1 junction}). Clearly it is irrelevant for applications but is useful to fix a few ideas that we will rely upon in the subsequent theory.

\subsubsection{1-1 junction: The fictitious junction}
\label{sect:1-1_junction}
\begin{figure}[!t]
\centering
\includegraphics[width=0.6\textwidth]{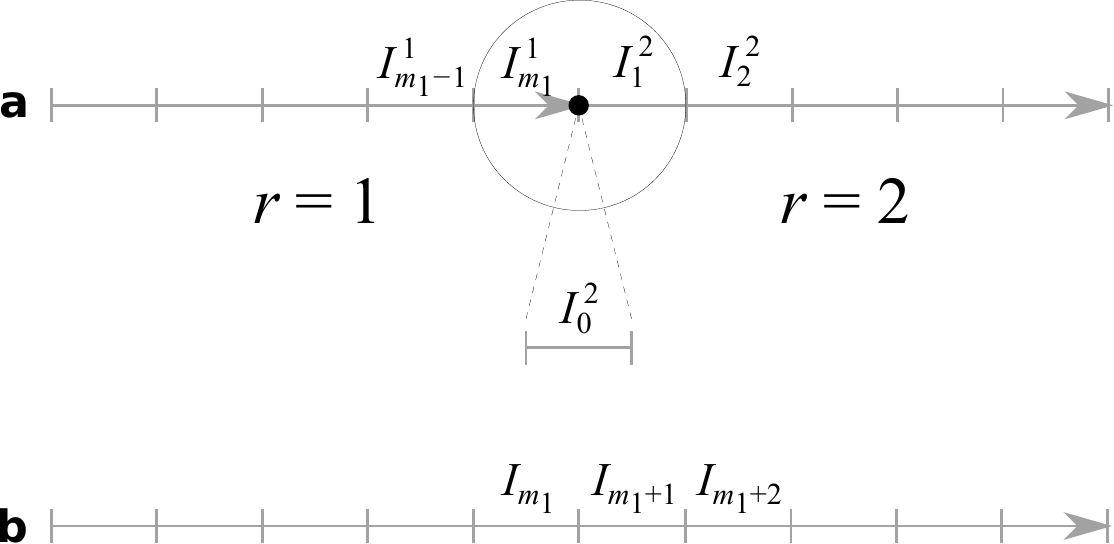}
\caption{(a) A 1-1 junction; (b) the corresponding interpretation as a single road}
\label{fig:1-1_junction}
\end{figure}
A 1-1 junction allows us to exploit the topological equivalence with a single road without junction, cf. Fig.~\ref{fig:1-1_junction}, for grasping some general ideas on how to handle transmission conditions in a kinetic model.

Let $r=1$ be the incoming road and $r=2$ the outgoing road. Vehicles are conserved through the junction if the net time variation of the total mass within the circle in Fig.~\ref{fig:1-1_junction}a is due only to inflow from the cell $I^1_{m_1-1}$ and outflow from the cell $I^2_1$. In order to impose such a constraint we consider Eq.~\eqref{eq:single.road.eq} for each of the road cells encompassed by the circle, i.e., $i=m_1$ for $r=1$ and $i=1$ for $r=2$. First we sum over $j=1,\,\dots,\,n$ taking Eq.~\eqref{eq:conservativeness} into account:
$$
\begin{cases}
	\dfrac{d\rho^1_{m_1}}{dt}+\Phi^1_{m_1,m_1+1}q^1_{m_1}-\Phi^1_{m_1-1,m_1}q^1_{m_1-1}=0 \\[3mm]
	\dfrac{d\rho^2_1}{dt}+\Phi^2_{1,2}q^2_1-\Phi^2_{0,1}q^2_0=0,
\end{cases}
$$
then we further sum term by term to find
\begin{align*}
	\frac{d}{dt}\left(\rho^1_{m_1}+\rho^2_1\right) &= \Phi^1_{m_1-1,m_1}q^1_{m_1-1}-\Phi^2_{1,2}q^2_1 \\
	& \phantom{=} +\left(\Phi^2_{0,1}q^2_0-\Phi^1_{m_1,m_1+1}q^1_{m_1}\right).
\end{align*}
By inspecting the right-hand side we recognize that the quantity in brackets should vanish:
\begin{equation}
	\Phi^1_{m_1,m_1+1}q^1_{m_1}=\Phi^2_{0,1}q^2_0,
	\label{eq:masscons_1-1}
\end{equation}
which, thinking of the fluxes in terms of the distribution functions, cf. Eq.~\eqref{eq:qcell}, yields a constraint on $\Phi^1_{m_1,m_1+1}$, $\{f^2_{0j}\}_{j=1}^{n}$ to be imposed at the junction.

Equation~\eqref{eq:masscons_1-1} alone is not sufficient to fully determine the transmission conditions as it involves $n+1$ unknown quantities. On the other hand, it is plain that, for conservation purposes, the entire flux reaching the end of road $r=1$ must be put at the entrance of road $r=2$, namely in the ghost cell $I^2_0$ (see Fig.~\ref{fig:1-1_junction}a). Hence we complement Eq.~\eqref{eq:masscons_1-1} with the rule
\begin{equation}
	q^2_0=q^1_{m_1},
	\label{eq:fluxdistr_1-1}
\end{equation}
which, considering that Eq.~\eqref{eq:masscons_1-1} is required to hold for all possible values of the incoming flux $q^1_{m_1}$, gives
\begin{equation}
	\Phi^1_{m_1,m_1+1}=\Phi^2_{0,1}.
	\label{eq:Phi_1-1}
\end{equation}

In addition, from Eq.~\eqref{eq:fluxdistr_1-1} we get the relationship
\begin{equation}
	\sum_{j=1}^{n}v_j\left(f^2_{0j}-f^1_{m_1j}\right)=0,
	\label{eq:fluxdistr_1-1_micro}
\end{equation}
which we need to take into account when assigning the distribution functions $\{f^2_{0j}\}_{j=1}^{n}$ at the junction. Notice that there is not, in principle, a unique way to choose the $f^2_{0j}$'s so as to fulfill Eq.~\eqref{eq:fluxdistr_1-1_micro}. This is because the flux distribution rule~\eqref{eq:fluxdistr_1-1} is macroscopic, thus it ignores the vehicle distribution over the microscopic speed classes.

\begin{remark}
An obvious way to satisfy Eq.~\eqref{eq:fluxdistr_1-1_micro} is to choose $f^2_{0j}=f^1_{m_1j}$ for all $j$. A less obvious one is $f^2_{0j}=0$ for $j=1,\,\dots,\,n-1$, $f^2_{0n}=q^1_{m_1}$.
\end{remark}

In order to fix a criterion for selecting the $f^2_{0j}$'s, which can possibly act as a paradigm for studying also the next types of junctions, we take inspiration from the topological equivalence between:
\begin{inparaenum}[(a)]
\item two roads connected by a 1-1 junction;
\item a single road without junction. 
\end{inparaenum}
Following Fig.~\ref{fig:1-1_junction}, we identify the respective road cells near the junction as $I^1_{m_1}\leftrightarrow I_{m_1}$, $I^2_1\leftrightarrow I_{m_1+1}$, $I^2_2\leftrightarrow I_{m_1+2}$ and write the solution across the junction in the two cases by integrating Eq.~\eqref{eq:single.road.eq} up to a time $t>0$:
\begin{align*}
	&\text{(a)}\ f^2_{1j}(t)=f^2_{1j}(0)+\displaystyle{\int_0^t}\Bigl\{v_j\left(\Phi^2_{0,1}f_{0j}^2-\Phi^2_{1,2}f^2_{1j}\right)
		+G_j[\f^2_1,\,\f^2_1]-f^2_{1j}L[\f^2_1]\Bigr\}\,ds \\
	&\text{(b)}\ f_{m_1+1,j}(t)=f_{m_1+1,j}(0)+\displaystyle{\int_0^t}\Bigl\{v_j\left(\Phi_{m_1,m_1+1}f_{m_1j}-\Phi_{m_1+1,m_1+2}f_{m_1+1,j}\right) \\
		& \qquad\qquad\qquad\qquad\qquad\qquad\qquad +G_j[\f_{m_1+1},\,\f_{m_1+1}]-f_{m_1+1,j}L[\f_{m_1+1}]\Bigr\}\,ds.
\end{align*}
Imposing, as it is natural, that the two solutions coincide, i.e.,
$$ f^1_{m_1j}=f_{m_1j}, \quad f^2_{1j}=f_{m_1+1,j}, \quad f^2_{2j}=f_{m_1+2,j} \qquad \forall\,j, $$
implies in particular $\rho^2_1=\rho_{m_1+1}$, $\rho^2_2=\rho_{m_1+2}$, hence, according to Eq.~\eqref{eq:fluxlim}, $\Phi^2_{1,2}=\Phi_{m_1+1,m_1+2}$. Next, subtracting term by term yields
$$ \int_0^tv_j\left(\Phi^2_{0,1}f^2_{0j}-\Phi_{m_1,m_1+1}f^1_{m_1j}\right)\,ds=0 \qquad \forall\,j, $$
that is, for the arbitrariness of $t>0$ and owing to Eq.~\eqref{eq:Phi_1-1},
$$ v_j\left(\Phi^1_{m_1,m_1+1}f^2_{0j}-\Phi_{m_1,m_1+1}f^1_{m_1j}\right)=0 \qquad \forall\,j. $$
Finally, in order for roads $r=1,\,2$ in Fig.~\ref{fig:1-1_junction}a to be truly a unique road like in Fig.~\ref{fig:1-1_junction}b we must define $\Phi^1_{m_1,m_1+1}=\Phi_{m_1,m_1+1}$, so that ultimately
$$ v_j\Phi_{m_1,m_1+1}\left(f^2_{0j}-f^1_{m_1j}\right)=0 \qquad \forall\,j. $$

This relationship gives $f^2_{0j}=f^1_{m_1j}$, which is indeed compatible with Eq.~\eqref{eq:fluxdistr_1-1_micro}. Notice however that such a condition is not tight for $j=1$ since $v_1=0$. In other words, it is possible to choose $f^2_{01}\ne f^1_{m_11}$ while still preserving the equivalence between cases (a) and (b) of Fig.~\ref{fig:1-1_junction}. Considering that motionless vehicles possibly present in the cell $I^1_{m_1}$ cannot cross the junction by definition, an alternative choice, which will turn out to be more convenient for the subsequent theory (cf. Remark~\ref{rem:dens-flux_est} at the end of Section~\ref{sect:2-1_junction}), is $f^2_{01}=0$.

On the whole, the transmission conditions
$$
f^2_{0j}=
\begin{cases}
	0 & \text{if\ } j=1 \\
	f^1_{m_1j} & \text{if\ } j\geq 2
\end{cases}
$$
mean that vehicles do not change speed class when crossing the junction. Indeed the speed distribution of moving vehicles beyond the junction is the same as that before the junction.

\subsubsection{1-2 junction: The flux distribution rule}
\label{sect:1-2_junction}
\begin{figure}[!t]
\centering
\includegraphics[width=0.5\textwidth]{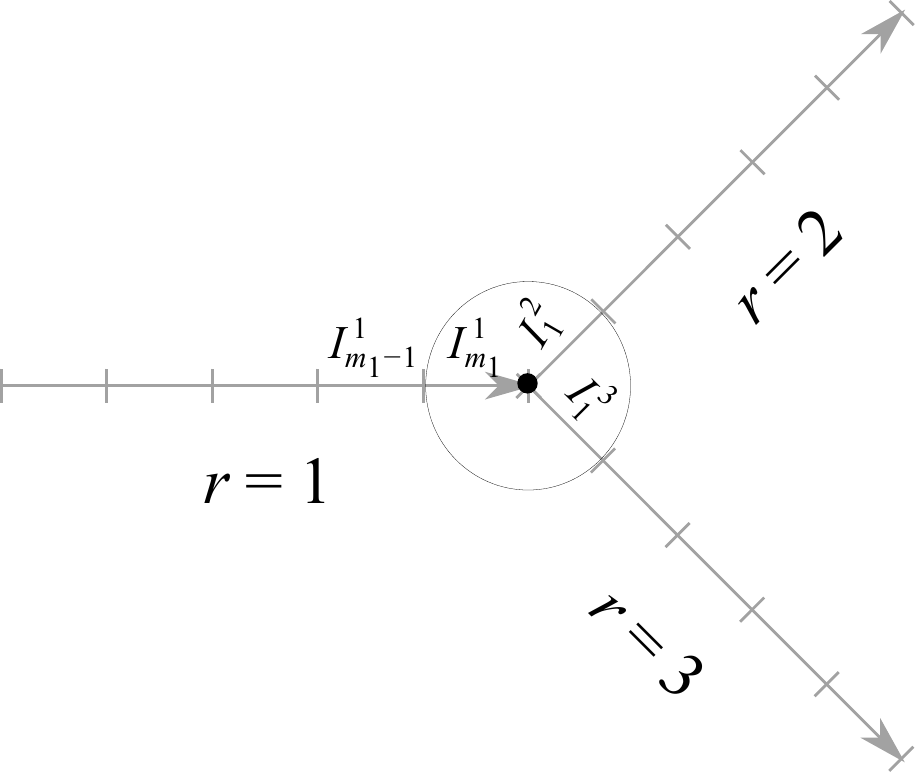}
\caption{Prototype of 1-2 junction}
\label{fig:1-2_junction}
\end{figure}
Let us consider the 1-2 junction illustrated in Fig.~\ref{fig:1-2_junction}, where $r=1$ is the incoming road and $r=2,\,3$ are the outgoing roads. Writing Eq.~\eqref{eq:single.road.eq} in the road cells encompassed by the circle, then summing over $j=1,\,\dots,\,n$, and taking Eq.~\eqref{eq:conservativeness} into account gives:
$$
\begin{cases}
	\dfrac{d\rho^1_{m_1}}{dt}+\Phi^1_{m_1,m_1+1}q^1_{m_1}-\Phi^1_{m_1-1,m_1}q^1_{m_1-1}=0 \\[3mm]
	\dfrac{d\rho^2_1}{dt}+\Phi^2_{1,2}q^2_1-\Phi^2_{0,1}q^2_0=0 \\[3mm]
	\dfrac{d\rho^3_1}{dt}+\Phi^3_{1,2}q^3_1-\Phi^3_{0,1}q^3_0=0,
\end{cases}
$$
whence, summing further term by term,
\begin{align*}
	\frac{d}{dt}(\rho^1_{m_1}+\rho^2_1+\rho^3_1) &= \Phi^1_{m_1-1,m_1}q^1_{m_1-1}-\Phi^2_{1,2}q^2_1-\Phi^3_{1,2}q^3_1 \\
	& \phantom{=}+\left(\Phi^2_{0,1}q^2_0+\Phi^3_{0,1}q^3_0-\Phi^1_{m_1,m_1+1}q^1_{m_1}\right).
\end{align*}
The mass conservation principle holds at the junction provided the quantity in brackets is zero:
\begin{equation}
	\Phi^1_{m_1,m_1+1}q^1_{m_1}=\Phi^2_{0,1}q^2_0+\Phi^3_{0,1}q^3_0,
	\label{eq:masscons_1-2}
\end{equation}
which gives a constraint on $\Phi^1_{m_1,m_1+1}$, $\{f^2_{0j}\}_{j=1}^{n}$, $\{f^3_{0j}\}_{j=1}^{n}$.

In order to characterize the dynamics at the junction, Eq.~\eqref{eq:masscons_1-2} has to be complemented with a model of the flow splitting occurring when incoming vehicles from road $r=1$ redistribute on the outgoing roads $r=2,\,3$. At a macroscopic level, such a flow splitting is classically expressed as follows:
\begin{trafficrule}{Flux distribution}
There exists a number $a\in[0,\,1]$, called the \emph{flux distribution coefficient}, such that
\begin{equation}
	q^2_0=aq^1_{m_1}, \qquad q^3_0=(1-a)q^1_{m_1}.
	\label{eq:fluxdistr}
\end{equation}
\end{trafficrule}
In practice the coefficient $a$ is the percentage of vehicles which continue on road $r=2$ after leaving road $r=1$. Complementally, $1-a$ is the percentage of those which continue on road $r=3$. Notice that this rule is inspired by the conservation of the mass of cars, in fact it implies $q^1_{m_1}=q^2_0+q^3_0$, namely the continuity of the free flux at the junction.

Equation~\eqref{eq:fluxdistr} yields
$$ \sum_{j=1}^{n}v_j\left(f^2_{0j}-af^1_{m_1j}\right)=0, \qquad \sum_{j=1}^{n}v_j\left(f^3_{0j}-(1-a)f^1_{m_1j}\right)=0, $$
which, guided by the criterion established in Section~\ref{sect:1-1_junction}, we fulfill as
\begin{equation}
f^2_{0j}=
\begin{cases}
	0 & \text{if\ } j=1 \\
	af^1_{m_1j} & \text{if\ } j\geq 2,
\end{cases}
\qquad
f^3_{0j}=
\begin{cases}
	0 & \text{if\ } j=1 \\
	(1-a)f^1_{m_1j} & \text{if\ } j\geq 2.
\end{cases}
\label{eq:f0j_1-2}
\end{equation}

Finally, plugging Eq.~\eqref{eq:fluxdistr} into Eq.~\eqref{eq:masscons_1-2} and owing to the arbitrariness of the incoming flux $q^1_{m_1}$ we deduce
\begin{equation}
	\Phi^1_{m_1,m_1+1}=a\Phi^2_{0,1}+(1-a)\Phi^3_{0,1}.
	\label{eq:Phi_1-2}
\end{equation}

On the whole, Eqs.~\eqref{eq:f0j_1-2},~\eqref{eq:Phi_1-2} completely define an admissible set of transmission conditions for a 1-2 junction.

\begin{remark}
For $a=0,\,1$ the 1-2 junction reduces to a 1-1 junction plus a disconnected road. The choice~\eqref{eq:f0j_1-2} guarantees that the corresponding solution is the same as that of the topologically equivalent case of a single road without junction, cf. Section~\ref{sect:1-1_junction}.
\end{remark}

\subsubsection{2-1 junction: The right-of-way rule}
\label{sect:2-1_junction}
\begin{figure}[!t]
\centering
\includegraphics[width=0.5\textwidth]{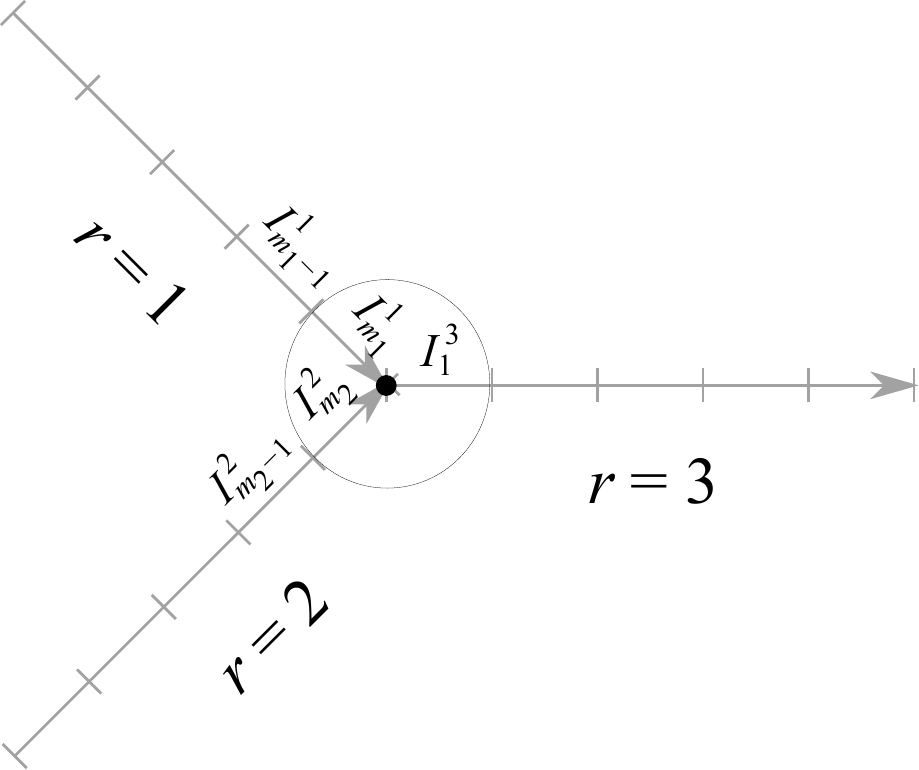}
\caption{Prototype of 2-1 junction}
\label{fig:2-1_junction}
\end{figure}
Let us consider now the 2-1 junction illustrated in Fig.~\ref{fig:2-1_junction}, where $r=1,\,2$ are the incoming roads and $r=3$ is the outgoing road. In order to enforce the mass conservation constraint at the junction we write Eq.~\eqref{eq:single.road.eq} in the road cells encompassed by the circle, summing over $j=1,\,\dots,\,n$ and taking Eq.~\eqref{eq:conservativeness} into account:
$$
\begin{cases}
	\dfrac{d\rho^1_{m_1}}{dt}+\Phi^1_{m_1,m_1+1}q^1_{m_1}-\Phi^1_{m_1-1,m_1}q^1_{m_1-1}=0 \\[3mm]
	\dfrac{d\rho^2_{m_2}}{dt}+\Phi^2_{m_2,m_2+1}q^2_{m_2}-\Phi^2_{m_2-1,m_2}q^2_{m_2-1}=0 \\[3mm]
	\dfrac{d\rho^3_1}{dt}+\Phi^3_{1,2}q^3_1-\Phi^3_{0,1}q^3_0=0.
\end{cases}
$$
Then we further sum term by term to discover
\begin{align*}
	\frac{d}{dt}(\rho^1_{m_1}+\rho^2_{m_2}+\rho^3_1) &= \Phi^1_{m_1-1,m_1}q^1_{m_1-1}+\Phi^2_{m_2-1,m_2}q^2_{m_2-1}-\Phi^3_{1,2}q^3_1 \\
	& \phantom{=}+\left(\Phi^3_{0,1}q^3_0-\Phi^1_{m_1,m_1+1}q^1_{m_1}-\Phi^2_{m_2,m_2+1}q^2_{m_2}\right)
\end{align*}
and we set to zero the quantity in brackets at the right-hand side:
\begin{equation}
	\Phi^1_{m_1,m_1+1}q^1_{m_1}+\Phi^2_{m_2,m_2+1}q^2_{m_2}=\Phi^3_{0,1}q^3_0.
	\label{eq:masscons_2-1}
\end{equation}
This equation formalizes the constraint on $\Phi^1_{m_1,m_1+1}$, $\Phi^2_{m_2,m_2+1}$, $\{f^3_{0j}\}_{j=1}^{n}$ to be imposed at the junction in this case.

Incoming fluxes from roads $r=1,\,2$ merge at the junction for entering road $r=3$. However this may not be always possible, because the outgoing flux cannot in any case exceed the maximum admissible flux on road $r=3$ ($q^3_{0}\leq 1$) nor can the car density at the entrance of road $r=3$ exceed the maximum admissible density in a single road cell ($\rho^3_0\leq 1$). Consequently, it is necessary to devise a model for the flow merging based on a criterion which establishes which one of the incoming roads has the right-of-way in case that their vehicles cannot flow simultaneously into road $r=3$.

Let $r^\ast\in\{1,\,2\}$ be the incoming road with right-of-way and $r_\ast\in\{1,\,2\}\setminus\{r^\ast\}$ the other incoming road. Moreover, let $p\in[0,\,1]$ be a flux threshold (to be determined) which triggers the right-of-way of road $r^\ast$ over road $r_\ast$. At a macroscopic level, a right-of-way rule can be formulated in the following way:
\begin{trafficrule}{Right-of-way}
\begin{itemize}
\item If $q^1_{m_1}+q^2_{m_2}\leq p$ then
\begin{equation}
	q^3_0=q^1_{m_1}+q^2_{m_2}.
	\label{eq:rightofway_1}
\end{equation}
\item Conversely, if $q^1_{m_1}+q^2_{m_2}>p$ then
\begin{equation}
\begin{cases}
	q^3_0=q^{r^\ast}_{m_{r^\ast}} \\
	\Phi^{r_\ast}_{m_{r_\ast},m_{r_\ast}+1}=0.
\end{cases}
\label{eq:rightofway_2}
\end{equation}
\end{itemize}
\end{trafficrule}

This rule states that if the total incoming flux does not exceed the threshold $p$ then cars from both roads $r=1,\,2$ flow simultaneously into road $r=3$. Otherwise only cars from road $r^\ast$ are allowed to flow into road $r=3$, cars from road $r_\ast$ being forced to stop (this is indeed the effect of setting to zero the flux limiter at the end of road $r_\ast$). The rule is again inspired by the continuity of the free flux across the junction. This is evident for Eq.~\eqref{eq:rightofway_1}, whereas for Eq.~\eqref{eq:rightofway_2} it suffices to consider that the condition imposed on the flux limiter $\Phi^{r_\ast}_{m_{r_\ast},m_{r_\ast}+1}$ disconnects road $r_\ast$ from the junction, thereby making road $r=3$ the continuation of road $r^\ast$.

Owing to the rule above we get
$$
\begin{cases}
	\displaystyle{\sum_{j=1}^{n}}v_j\left(f^3_{0j}-f^1_{m_1j}-f^2_{m_2j}\right)=0 & \text{if\ } q^1_{m_1}+q^2_{m_2}\leq p \\
	\displaystyle{\sum_{j=1}^{n}}v_j\left(f^3_{0j}-f^{r^\ast}_{m_{r^\ast}j}\right)=0 & \text{if\ } q^1_{m_1}+q^2_{m_2}>p,
\end{cases}
$$
which, appealing again to the criterion discussed in Section~\ref{sect:1-1_junction}, we fulfill as
\begin{equation}
f^3_{0j}=
\begin{cases}
	0 & \text{if\ } j=1 \\
	f^1_{m_1j}+f^2_{m_2j} & \text{if\ } j\geq 2\ \text{and\ } q^1_{m_1}+q^2_{m_2}\leq p \\
	f^{r^\ast}_{m_{r^\ast}j} & \text{if\ } j\geq 2\ \text{and\ } q^1_{m_1}+q^2_{m_2}>p.
\end{cases}
\label{eq:f0j_2-1}
\end{equation}

Moreover, by plugging Eqs.~\eqref{eq:rightofway_1},~\eqref{eq:rightofway_2} into Eq.~\eqref{eq:masscons_2-1} and invoking the arbitrariness of the incoming fluxes $q^1_{m_1}$, $q^2_{m_2}$ we obtain the flux limiters at the end of the incoming roads:
\begin{itemize}
\item if $q^1_{m_1}+q^2_{m_2}\leq p$ then
\begin{equation}
\begin{cases}
	\Phi^1_{m_1,m_1+1}=\Phi^3_{0,1} \\
	\Phi^2_{m_2,m_2+1}=\Phi^3_{0,1};
\end{cases}
\label{eq:Phi_2-1_1}
\end{equation}
\item if $q^1_{m_1}+q^2_{m_2}>p$ then
\begin{equation}
\begin{cases}
	\Phi^{r^\ast}_{m_{r^\ast},m_{r^\ast}+1}=\Phi^3_{0,1} \\
	\Phi^{r_\ast}_{m_{r_\ast},m_{r_\ast}+1}=0.
\end{cases}
\label{eq:Phi_2-1_2}
\end{equation}
\end{itemize}

Equations~\eqref{eq:f0j_2-1}--\eqref{eq:Phi_2-1_2} provide a complete set of transmission conditions for the 2-1 junction up to determining the parameter $p$. To this end we focus on the case $q^1_{m_1}+q^2_{m_2}\leq p$ noticing that, according to Eq.~\eqref{eq:f0j_2-1}, the density at the entrance of road $r=3$ is (cf. Eq.~\eqref{eq:rhocell}):
\begin{align*}
	\rho^3_0 &= \sum_{j=2}^{n}f^3_{0j}=\sum_{j=2}^{n}\left(f^1_{m_1j}+f^2_{m_2j}\right) & (\text{because\ } f^3_{01}=0) \\
	&\leq \frac{1}{v_2}\sum_{j=2}^{n}v_j\left(f^1_{m_1j}+f^2_{m_2j}\right) & (\text{because\ } 0<v_2\leq v_j\ \forall\,j\geq 2) \\
	&= \frac{q^1_{m_1}+q^2_{m_2}}{v_2} & (\text{because\ } v_1=0) \\
	&\leq \frac{p}{v_2}.
\end{align*}
As previously mentioned, the merging of the two flows is admissible if $p\leq 1$ and $\rho^3_0\leq 1$, which is guaranteed by choosing
$$ p\leq v_2, $$
for instance exactly $p=v_2$.

\begin{remark}
Here is the point where the choice $f^3_{01}=0$, as opposed to $f^3_{01}=f^1_{m_11}+f^2_{m_21}$, is fully exploited. In fact the above estimate of the density via the flux is ultimately possible thanks to the fact that the sum defining $\rho^3_0$ starts from $j=2$.
\label{rem:dens-flux_est}
\end{remark}

\begin{remark}
When $q^1_{m_1}+q^2_{m_2}>p$ the 2-1 junction reduces to a 1-1 junction plus a disconnected road. The choice~\eqref{eq:f0j_2-1} is such that the corresponding solution is, once again, the same as that of a single road without junction, cf. Section~\ref{sect:1-1_junction}.
\end{remark}

\section{Numerical tests}
\label{sect:numerics}
In this section we present a computational analysis of two exploratory case studies, which show the potential of our kinetic model to tackle realistic network problems of interest for applications.

We partition every road of the network in $m_r=6$ cells and we fix $n=6$ uniformly spaced speed classes, cf. Eq.~\eqref{eq:vj_uniform}. On each road, we integrate in time the differential equations~\eqref{eq:single.road.eq} by the classical fourth order explicit Runge-Kutta scheme. As already recalled in the Introduction, we take from~\cite{fermo2013SIAP} the expression of the table of games characterizing the gain operators~\eqref{eq:GL} (we do not report it here for the sake of conciseness).

\subsection{Traffic circle: The effect of right-of-way inversion}
\begin{figure}[!t]
\centering
\includegraphics[width=0.7\textwidth]{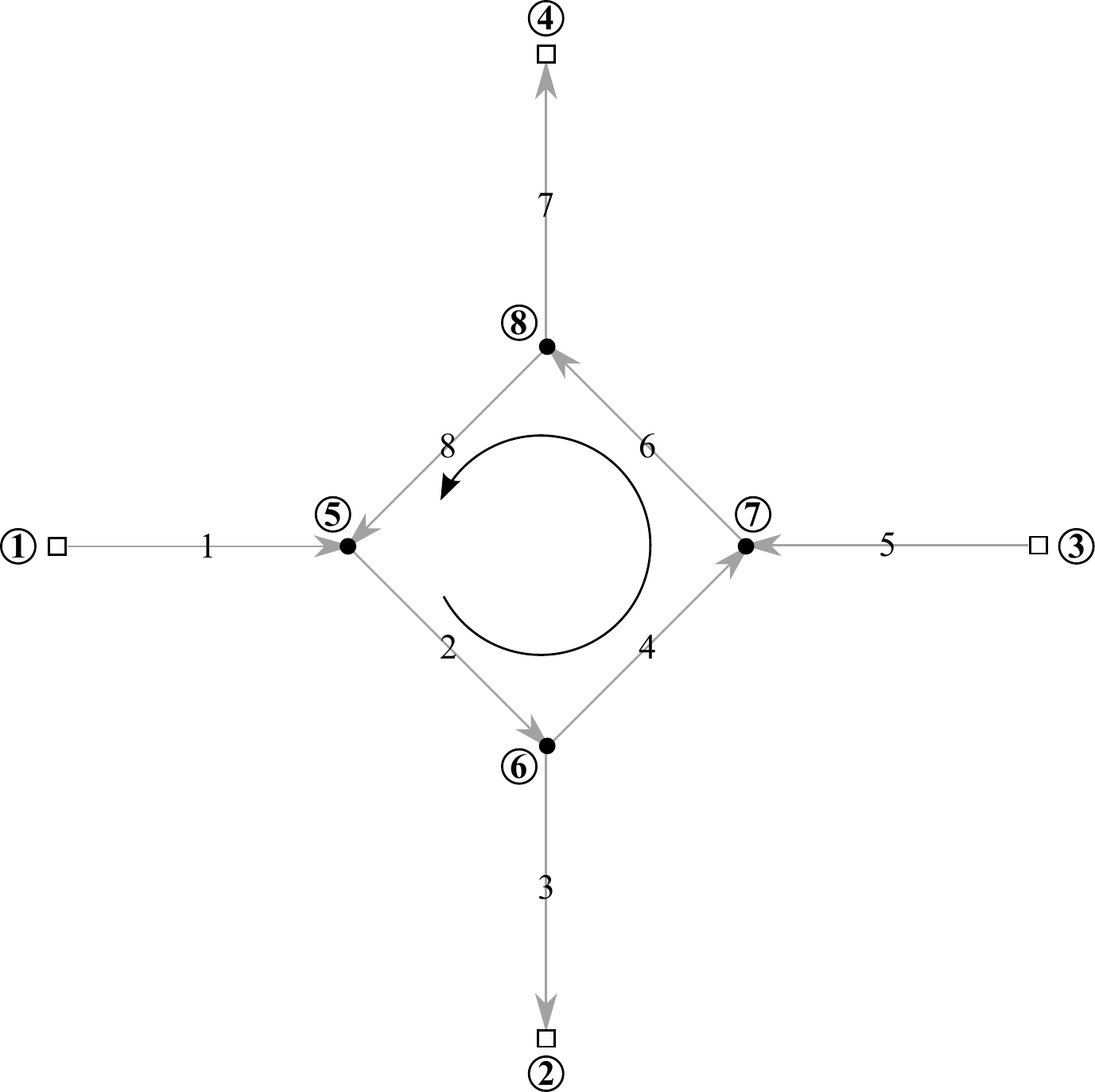}
\caption{A prototype of counterclockwise traffic circle (right-hand traffic)}
\label{fig:traffic_circle}
\end{figure}
The first case study concerns a prototype of traffic circle with counterclockwise circulation (right-hand traffic), which we model as a network of eight roads and as many vertexes, see Fig.~\ref{fig:traffic_circle}. Roads $1$, $5$ give access to the inner circle, whereas roads $3$, $7$ are the exits. The actual circle is formed by roads $2$, $4$, $6$, and $8$. 

Vertexes \ding{192} to \ding{195} are the peripheral ones. In particular, \ding{192}, \ding{194} at the beginning of roads $1$, $5$, respectively, are access points where boundary conditions $\{f^1_{0j}\}_{j=1}^{6}$, $\{f^5_{0j}\}_{j=1}^{6}$ have to be prescribed, whereas \ding{193}, \ding{195} at the end of roads $3$, $7$, respectively, are exit points where downstream conditions have to be implemented by specifying the flux limiters $\Phi^3_{6,7}$, $\Phi^7_{6,7}$. We invariably set
$$ f^1_{0j}=f^5_{0j}=
	\begin{cases}
		0 & \text{if\ } j\leq 5 \\
		1 & \text{if\ } j=6,
	\end{cases}
	\qquad\qquad
	\Phi^3_{6,7}=\Phi^7_{6,7}=1,
$$
thus we assume that vehicles access the network at the maximum density and speed and freely leave it thanks e.g., to the absence of traffic jams downstream.

Vertexes \ding{196} to \ding{199} are true junctions. In particular, \ding{197}, \ding{199} are 1-2 junctions calling for a flux distribution parameter, that we fix to $a=\frac{1}{2}$. Hence, according to the theory developed in Section~\ref{sect:1-2_junction}, we have the following transmission conditions:
\begin{itemize}
\item At junction \ding{197}:
$$ f^3_{0j}=f^4_{0j}=
	\begin{cases}
		0 & \text{if\ } j=1 \\
		\frac{1}{2}f^2_{6j} & \text{if\ } j\geq 2,
	\end{cases}
	\qquad\qquad
	\Phi^2_{6,7}=\frac{\Phi^3_{0,1}+\Phi^4_{0,1}}{2},
$$
where $\Phi^3_{0,1}$, $\Phi^4_{0,1}$ are defined according to Eq.~\eqref{eq:fluxlim}.
\item At junction \ding{199}:
$$ f^7_{0j}=f^8_{0j}=
	\begin{cases}
		0 & \text{if\ } j=1 \\
		\frac{1}{2}f^6_{6j} & \text{if\ } j\geq 2,
	\end{cases}
	\qquad\qquad
	\Phi^6_{6,7}=\frac{\Phi^7_{0,1}+\Phi^8_{0,1}}{2},
$$
where $\Phi^7_{0,1}$, $\Phi^8_{0,1}$ are in turn defined according to Eq.~\eqref{eq:fluxlim}.
\end{itemize}

Conversely, \ding{196}, \ding{198} are 2-1 junctions, which require the specification of a right-of-way rule. Following the theory developed in Section~\ref{sect:2-1_junction}, we set the flux threshold to $p=v_2=\frac{1}{5}$ (cf. Eq.~\eqref{eq:vj_uniform} with $n=6$), then we need to decide which roads have the right-of-way at each junction. We consider three possible choices, aiming at comparing their global impact on the flow in the traffic circle.

\paragraph*{Case 1: Usual right-of-way}
The first choice consists in assuming that, coherently with usual traffic rules, vehicles entering the traffic circle give priority to the circulating flow. This means that at junction \ding{196} road $8$ has right-of-way over road $1$ and, likewise, at junction \ding{198} road $4$ has right-of-way over road $5$. Hence in this case transmission conditions are implemented as follows:
\begin{itemize}
\item At junction \ding{196}:
\begin{gather*}
	f^2_{0j}=
	\begin{cases}
		0 & \text{if\ } j=1 \\
		f^1_{6j}+f^8_{6j} & \text{if\ } j\geq 2,\ q^1_6+q^8_6\leq\frac{1}{5} \\
		f^8_{6j} & \text{if\ } j\geq 2,\ q^1_6+q^8_6>\frac{1}{5},
	\end{cases}
	\\[3mm]
	\begin{array}{rl}
		\left.
		\begin{array}{rcl}
			\Phi^1_{6,7} & = & \Phi^2_{0,1} \\
			\Phi^8_{6,7} & = & \Phi^2_{0,1}
		\end{array}
		\right\} & \text{if\ } q^1_6+q^8_6\leq\frac{1}{5}
		\\[5mm]
		\left.
		\begin{array}{rcl}
			\Phi^1_{6,7} & = & 0 \\
			\Phi^8_{6,7} & = & \Phi^2_{0,1}
		\end{array}
		\right\} & \text{if\ } q^1_6+q^8_6>\frac{1}{5},
	\end{array}
\end{gather*}
where $\Phi^2_{0,1}$ is defined according to Eq.~\eqref{eq:fluxlim}.
\item At junction \ding{198}:
\begin{gather*}
	f^6_{0j}=
	\begin{cases}
		0 & \text{if\ } j=1 \\
		f^4_{6j}+f^5_{6j} & \text{if\ } j\geq 2,\ q^4_6+q^5_6\leq\frac{1}{5} \\
		f^4_{6j} & \text{if\ } j\geq 2,\ q^4_6+q^5_6>\frac{1}{5},
	\end{cases}
	\\[3mm]
	\begin{array}{rl}
		\left.
		\begin{array}{rcl}
			\Phi^4_{6,7} & = & \Phi^6_{0,1} \\
			\Phi^5_{6,7} & = & \Phi^6_{0,1}
		\end{array}
		\right\} & \text{if\ } q^4_6+q^5_6\leq\frac{1}{5}
		\\[5mm]
		\left.
		\begin{array}{rcl}
			\Phi^4_{6,7} & = & \Phi^6_{0,1} \\
			\Phi^5_{6,7} & = & 0
		\end{array}
		\right\} & \text{if\ } q^4_6+q^5_6>\frac{1}{5},
	\end{array}
\end{gather*}
where $\Phi^6_{0,1}$ is defined according to Eq.~\eqref{eq:fluxlim} as well.
\end{itemize}

\paragraph*{Case 2: Inverted right-of-way}
The second choice consists in inverting the rule above, assuming that, contrary to the usual traffic rules, vehicles circulating in the traffic circle give priority to those entering it. Hence at junction \ding{196} road $1$ has, in this case, right-of-way over road $8$ and, similarly, at junction \ding{198} road $5$ has right-of-way over road $4$. Transmission conditions modify consequently as:
\begin{itemize}
\item At junction \ding{196}:
\begin{gather*}
	f^2_{0j}=
	\begin{cases}
		0 & \text{if\ } j=1 \\
		f^1_{6j}+f^8_{6j} & \text{if\ } j\geq 2,\ q^1_6+q^8_6\leq\frac{1}{5} \\
		f^1_{6j} & \text{if\ } j\geq 2,\ q^1_6+q^8_6>\frac{1}{5},
	\end{cases}
	\\[3mm]
	\begin{array}{rl}
		\left.
		\begin{array}{rcl}
			\Phi^1_{6,7} & = & \Phi^2_{0,1} \\
			\Phi^8_{6,7} & = & \Phi^2_{0,1}
		\end{array}
		\right\} & \text{if\ } q^1_6+q^8_6\leq\frac{1}{5}
		\\[5mm]
		\left.
		\begin{array}{rcl}
			\Phi^1_{6,7} & = & \Phi^2_{0,1} \\
			\Phi^8_{6,7} & = & 0
		\end{array}
		\right\} & \text{if\ } q^1_6+q^8_6>\frac{1}{5};
	\end{array}
\end{gather*}
\item At junction \ding{198}:
\begin{gather*}
	f^6_{0j}=
	\begin{cases}
		0 & \text{if\ } j=1 \\
		f^4_{6j}+f^5_{6j} & \text{if\ } j\geq 2,\ q^4_6+q^5_6\leq\frac{1}{5} \\
		f^5_{6j} & \text{if\ } j\geq 2,\ q^4_6+q^5_6>\frac{1}{5},
	\end{cases}
	\\[5mm]
	\begin{array}{rl}
		\left.
		\begin{array}{rcl}
			\Phi^4_{6,7} & = & \Phi^6_{0,1} \\
			\Phi^5_{6,7} & = & \Phi^6_{0,1}
		\end{array}
		\right\} & \text{if\ } q^4_6+q^5_6\leq\frac{1}{5}
		\\[5mm]
		\left.
		\begin{array}{rcl}
			\Phi^4_{6,7} & = & 0 \\
			\Phi^5_{6,7} & = & \Phi^6_{0,1}
		\end{array}
		\right\} & \text{if\ } q^4_6+q^5_6>\frac{1}{5},
	\end{array}
\end{gather*}
\end{itemize}
$\Phi^2_{0,1}$, $\Phi^6_{0,1}$ being defined like in Case 1.

\begin{figure}[!t]
\centering
\includegraphics[width=\textwidth]{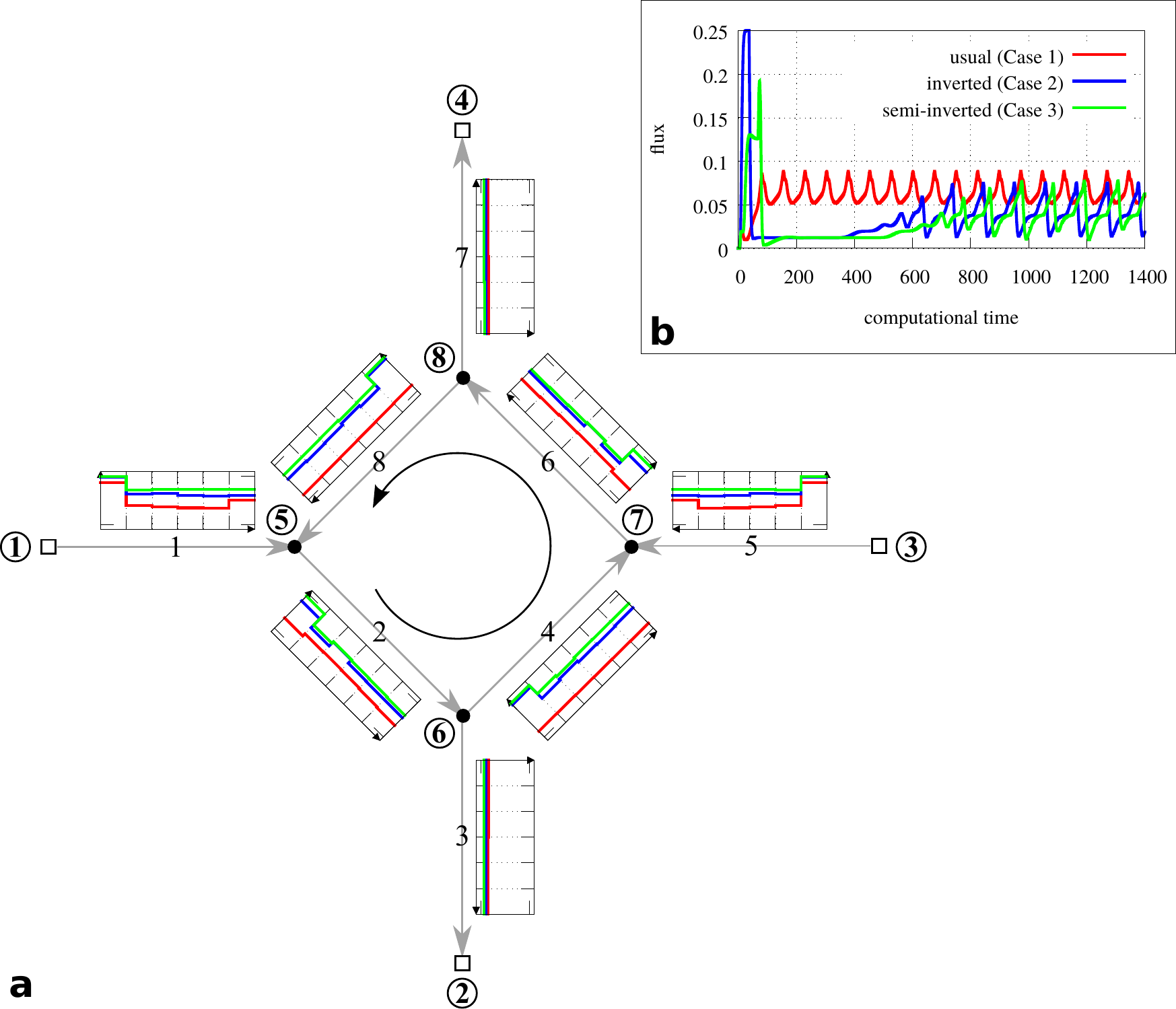}
\caption{(a) Instantaneous density profile in the network and (b) macroscopic flux at the beginning of roads 3, 7 for the problem of the traffic circle}
\label{fig:traffic_circle-sim}
\end{figure}

\paragraph*{Case 3: Semi-inverted right-of-way}
The third choice consists in assuming that the usual right-of-way rule is inverted at only one of the two 2-1 junctions, say e.g., junction \ding{198}. This can simulate the case in which such an inversion does not occur as a consequence of imposed traffic rules but is rather due to an incorrect behavior of drivers at that particular junction. Therefore transmission conditions at junction \ding{196} are like in Case 1, while those at junction \ding{198} are like in Case 2.
\medskip

Figure~\ref{fig:traffic_circle-sim}(a) shows the instantaneous profile of the vehicle density along the eight roads of the network. Red, blue, and green lines correspond to the aforesaid Cases 1, 2, and 3, respectively. Additionally, the top right panel in Fig.~\ref{fig:traffic_circle-sim}(b) displays the time trend of the macroscopic flux in the first cell ($i=1$) of roads 3, 7, i.e., immediately at the exit of the traffic circle. Because of the symmetry of the problem, it results $q^3_1=q^7_1$.

The simulation clearly indicates that the usual right-of-way rule at junctions \ding{196}, \ding{198} (Case 1, red) gives rise, in the long run, to a lower congestion in the circle. Moreover, the outgoing flux of vehicles settles on a periodic oscillatory trend with maximum peak and mean values (apart from a short transient initial period). The inverted right-of-way rule at junctions \ding{196}, \ding{198} (Case 2, blue) causes instead a consistent rise of congestion in the circle followed by a remarkable outflow drop in roads 3, 7. Actually, the semi-inverted right-of-way rule (Case 3, green) demonstrates that a systematic violation of the usual priority at just one of these two junctions is sufficient by itself for the efficiency of the traffic circle to be severely compromised.

In conclusion, this case study confirms that the suitable right-of-way rule at a traffic circle consists in giving priority to the circulating flow, i.e., to vehicles already occupying the circle. In fact, this makes the overall car flow the most fluid one, which implies the most effective action of the traffic circle in redirecting vehicles from incoming to outgoing roads.

\subsection{Fork with internal link: The effect of road conditions}
\begin{figure}[!t]
\centering
\includegraphics[width=0.8\textwidth]{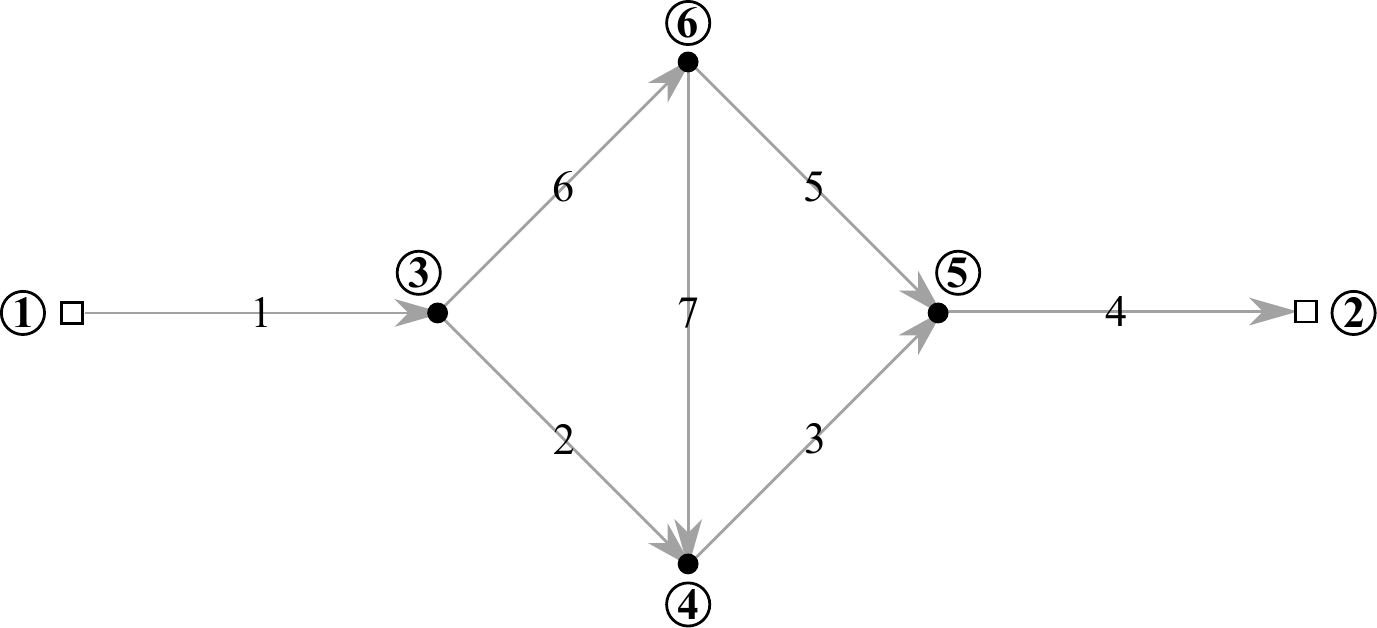}
\caption{A road fork with internal link (road 7)}
\label{fig:fork}
\end{figure}
The second case study concerns a road fork provided with an internal link, which allows vehicles to pass from one side of the fork to the other, see Fig.~\ref{fig:fork}. The goal of this example is to investigate, by means of the model developed in Section~\ref{sect:modeling}, the formation of congestions and the traffic fluidity downstream of the fork in connection with the actual practicability of such an internal link.

More precisely, with reference to Fig.~\ref{fig:fork}, we model various road conditions along edge 7 (the link) by means of a coefficient $\alpha\in [0,\,1]$, which parameterizes the probabilities of speed transition $A_{hk}^{j}[\rho^7_i]$ on that road, cf. Eq.~\eqref{eq:GL}. We recall again that the detailed expression of such probabilities can be found in~\cite{fermo2013SIAP}. For the present purposes, it will be sufficient to say that values of $\alpha$ close to $0$ stand for uncomfortable road conditions (such as e.g., scarce visibility, uneven road pavement, narrow roadway), which induce low speeds and little inclination of drivers to overtake. On the contrary, values of $\alpha$ close to $1$ stand for comfortable road conditions (such as e.g., good visibility, smooth road pavement, wide roadway), which give drivers more ease of maneuver.

At the peripheral vertexes \ding{192}, \ding{193} we impose the same access and downstream conditions as in the test of the traffic circle, namely:
$$ f^1_{0j}=
	\begin{cases}
		0 & \text{if\ } j\leq 5 \\
		1 & \text{if\ } j=6,
	\end{cases}
	\qquad\qquad
	\Phi^4_{6,7}=1,
$$
so that vehicles enter the network from road 1 at the maximum density and speed and freely leave it at the end of road 4.

At the 1-2 junctions \ding{194}, \ding{197} we fix the flux distribution parameter to $a=\frac{1}{2}$, hence from Section~\ref{sect:1-2_junction} we have:
$$ f^2_{0j}=f^6_{0j}=
	\begin{cases}
		0 & \text{if\ } j=1 \\
		\frac{1}{2}f^1_{6j} & \text{if\ } j\geq 2,
	\end{cases}
	\qquad\qquad
	\Phi^1_{6,7}=\frac{\Phi^2_{0,1}+\Phi^6_{0,1}}{2}
$$
and likewise
$$ f^5_{0j}=f^7_{0j}=
	\begin{cases}
		0 & \text{if\ } j=1 \\
		\frac{1}{2}f^6_{6j} & \text{if\ } j\geq 2,
	\end{cases}
	\qquad\qquad
	\Phi^6_{6,7}=\frac{\Phi^5_{0,1}+\Phi^7_{0,1}}{2},
$$
the flux limiters $\Phi^2_{0,1}$, $\Phi^5_{0,1}$, $\Phi^6_{0,1}$, $\Phi^7_{0,1}$ being given by Eq.~\eqref{eq:fluxlim}.

Finally, at the 2-1 junctions \ding{195}, \ding{196} we fix the flux threshold to $p=v_2=\frac{1}{5}$ (as it results from Eq.~\eqref{eq:vj_uniform} with $n=6$) and assume that roads 2, 3 have right-of-way over roads 7, 5, respectively (right-hand traffic). Thus from Section~\ref{sect:2-1_junction} we have:
\begin{gather*}
	f^3_{0j}=
	\begin{cases}
		0 & \text{if\ } j=1 \\
		f^2_{6j}+f^7_{6j} & \text{if\ } j\geq 2,\ q^2_6+q^7_6\leq\frac{1}{5} \\
		f^2_{6j} & \text{if\ } j\geq 2,\ q^2_6+q^7_6>\frac{1}{5}
	\end{cases}
	\intertext{with}
	\begin{array}{rl}
		\left.
		\begin{array}{rcl}
			\Phi^2_{6,7} & = & \Phi^3_{0,1} \\
			\Phi^7_{6,7} & = & \Phi^3_{0,1}
		\end{array}
		\right\} & \text{if\ } q^2_6+q^7_6\leq\frac{1}{5}
		\\[5mm]
		\left.
		\begin{array}{rcl}
			\Phi^2_{6,7} & = & \Phi^3_{0,1} \\
			\Phi^7_{6,7} & = & 0
		\end{array}
		\right\} & \text{if\ } q^2_6+q^7_6>\frac{1}{5}
	\end{array}
\end{gather*}
and likewise
\begin{gather*}
	f^4_{0j}=
	\begin{cases}
		0 & \text{if\ } j=1 \\
		f^3_{6j}+f^5_{6j} & \text{if\ } j\geq 2,\ q^3_6+q^5_6\leq\frac{1}{5} \\
		f^3_{6j} & \text{if\ } j\geq 2,\ q^3_6+q^5_6>\frac{1}{5}
	\end{cases}
	\intertext{with}
	\begin{array}{rl}
		\left.
		\begin{array}{rcl}
			\Phi^3_{6,7} & = & \Phi^4_{0,1} \\
			\Phi^5_{6,7} & = & \Phi^4_{0,1}
		\end{array}
		\right\} & \text{if\ } q^3_6+q^5_6\leq\frac{1}{5}
		\\[5mm]
		\left.
		\begin{array}{rcl}
			\Phi^3_{6,7} & = & \Phi^4_{0,1} \\
			\Phi^5_{6,7} & = & 0
		\end{array}
		\right\} & \text{if\ } q^3_6+q^5_6>\frac{1}{5},
	\end{array}
\end{gather*}
the flux limiters $\Phi^3_{0,1}$ and $\Phi^4_{0,1}$ being given by Eq.~\eqref{eq:fluxlim}.
\medskip

\begin{figure}[!t]
\centering
\includegraphics[width=\textwidth]{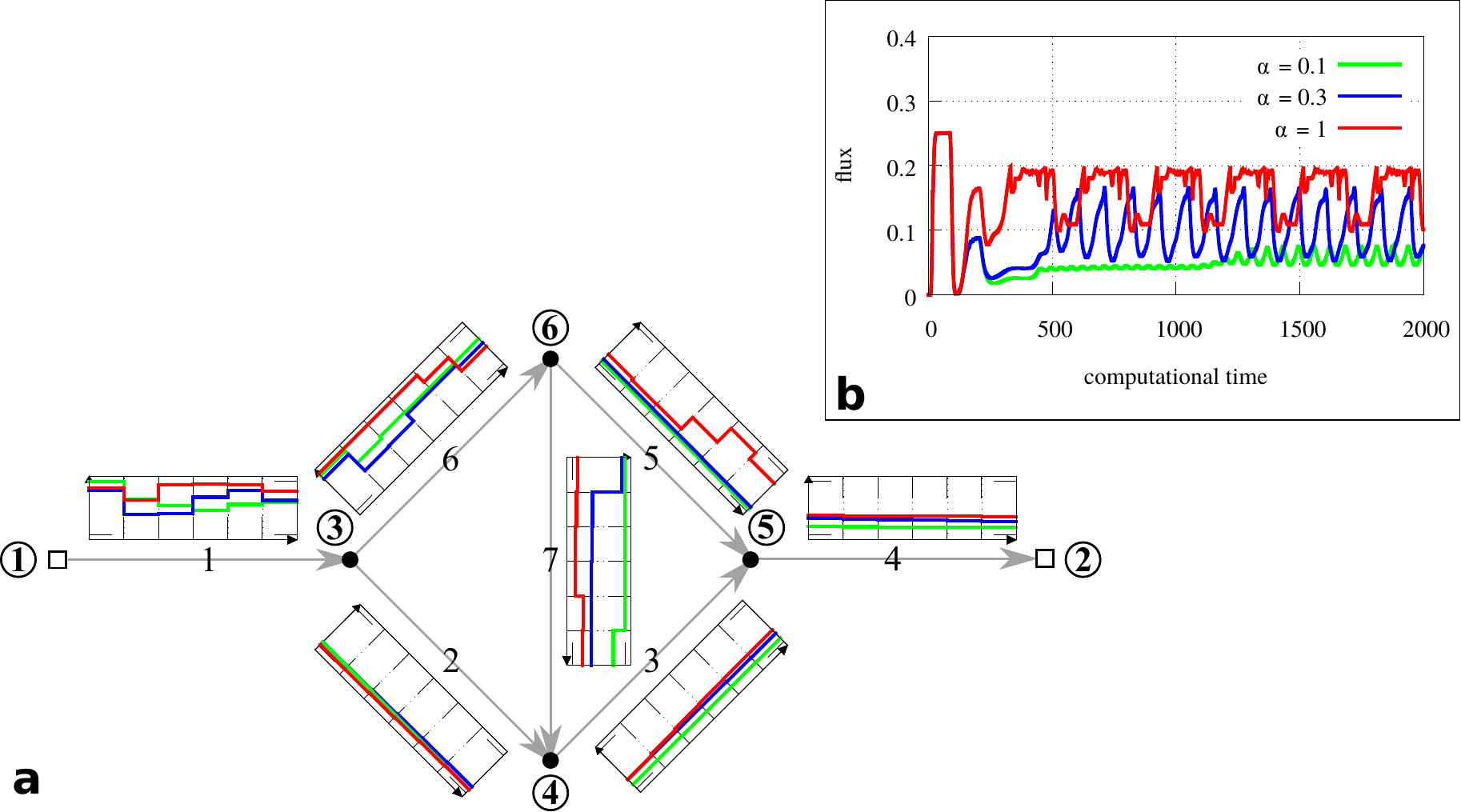}
\caption{(a) Instantaneous density profile in the network and (b) macroscopic flux at the beginning of road 4 for the problem of the fork}
\label{fig:fork-sim}
\end{figure}

Figure~\ref{fig:fork-sim}(a) shows the instantaneous car density profiles along the various roads of the network for three different values of the parameter $\alpha$ in road 7: $\alpha=1$ (optimal road conditions, red line), $\alpha=0.3$ (bad road conditions, blue line), $\alpha=0.1$ (very bad road conditions, green line). The simulation confirms that the worse the road conditions the higher the level of congestion in the link, where traffic is slowed down. The situation is reverted in roads 3, 4, 5 downstream of the fork, where instead better conditions of the upstream link increase the number of flowing cars, however without giving rise to congestions. In fact, the graph in Fig.~\ref{fig:fork}(b), in which the time trend of the macroscopic flux in the first cell of road 4 is plotted, suggests that, on average, improvements in the road conditions of the link enhance considerably the outgoing flux.

\section{Basic qualitative analysis} 
\label{sect:analysis}
In this section we study the well-posedness of the mathematical problems generated by the junction models introduced in Sections~\ref{sect:1-2_junction} and~\ref{sect:2-1_junction}, respectively.
In both cases we prove, under suitable assumptions, existence and uniqueness of the solution to the kinetic equations~\eqref{eq:single.road.eq} across the junction, as well as its continuous dependence on the data.

Let $X_{T}=C([0,\,T];\,\R^{n})$ be the Banach space of vector-valued continuous functions $t\mapsto \u(t)=(u_{1}(t),\,\dots,\,u_{n}(t)):[0,\,T]\to\R^{n}$. We define its closed subset
$$ \B=\left\{\u\in X_{T}\,:\, 0\leq u_j(t)\leq 1,\ \sum_{j=1}^n u_j(t)\leq 1,\ \forall\,j=1,\,\dots,\,n,\ \forall\,t\in[0,\,T]\right\}, $$
as well as the product space $\B^3=\B\times\B\times\B$, namely the set of functions $\U(t)=(\u^1(t),\u^2(t),\u^3(t))$ with $\u^r=(u^r_1,\,\dots,\,u^r_n)\in\B$ for $r=1,\,2,\,3$. We equip the latter with the $\infty$-norm:
$$ \norm{\U}_\infty:=\sup_{t\in[0,\,T]}\norm{\U(t)}_1 =
	\sup_{t\in[0,\,T]}\sum_{r=1}^{3}\norm{\u^r(t)}_1=\sup_{t\in[0,\,T]}\sum_{r=1}^{3}\sum_{j=1}^{n}\abs{u^r_{j}(t)}, $$
where $\norm{\cdot}_1$ is the classical 1-norm.

This is the general functional framework in which we will set our problems. Notice that the functions in $\B$ satisfy the non-negativity and boundedness required to the kinetic distribution functions $\f^r_i$ in every space cell $I_i$ of every road $r$. Thus the set $\B^3$ is the natural one where to look for solutions to the kinetic equations~\eqref{eq:single.road.eq} across junctions formed by three roads.

\subsection{Well-posedness of the 1-2 junction}
\label{sect:analysis_1-2}
\begin{figure}[!t]
\centering
\includegraphics[width=0.5\textwidth]{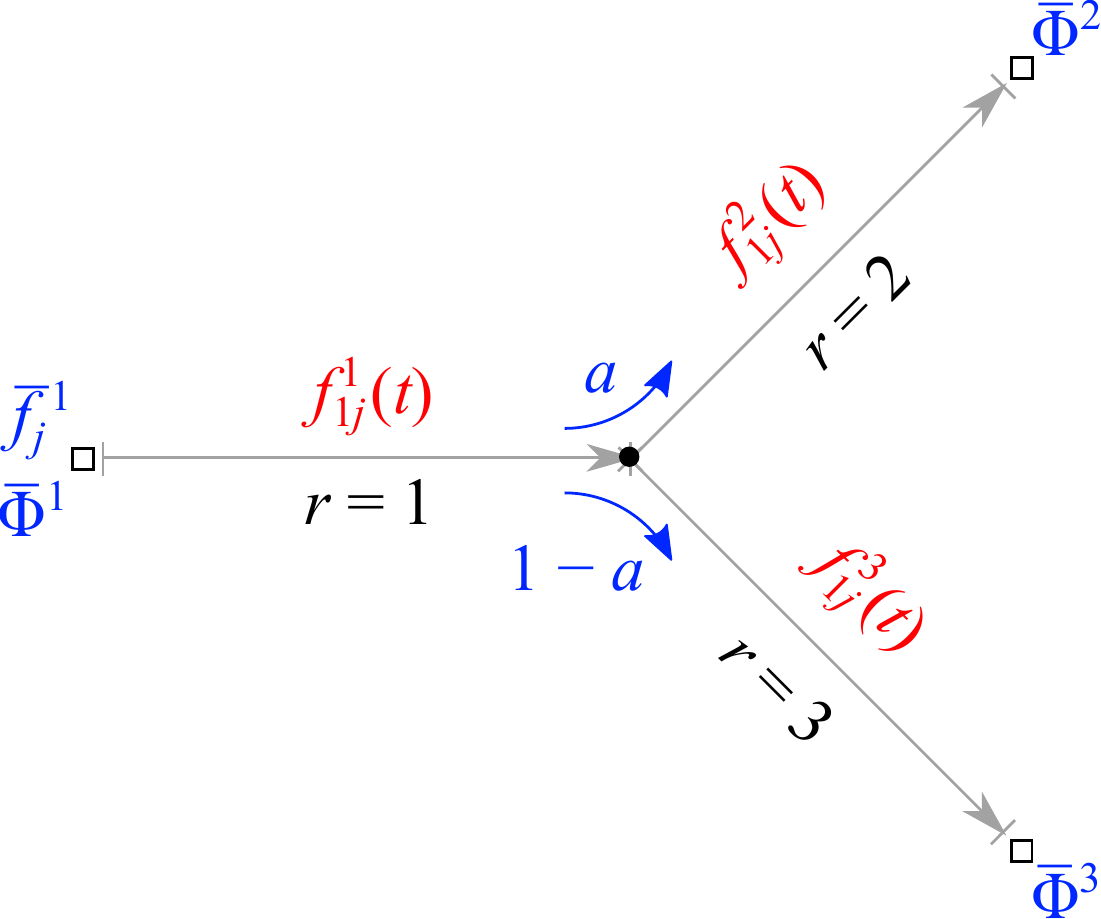}
\caption{Data (blue) and variables (red) for the problem of the 1-2 junction}
\label{fig:1-2_junction-analysis}
\end{figure}

Let us consider the 1-2 junction presented in Section~\ref{sect:1-2_junction} and for simplicity, but without loss of the generality, let us assume that each road consists of just one space cell, that is $m_r=1$ for $r=1,\,2,\,3$, see Fig.~\ref{fig:1-2_junction-analysis}. Introducing the symbol
$$ \deltac:=
	\begin{cases}
		0 & \text{if\ }	j=1 \\
		1 & \text{otherwise},
	\end{cases}
$$
which is a sort of ``complementary'' Kronecker delta, we can write compactly the transmission conditions~\eqref{eq:f0j_1-2},~\eqref{eq:Phi_1-2} as:
\begin{equation}
	f^2_{0j}=af^1_{1j}\deltac, \qquad f^3_{0j}=(1-a)f^1_{1j}\deltac,
		\qquad \Phi^1_{1,2}=a\Phi^2_{0,1}+(1-a)\Phi^3_{0,1}.
\label{eq:cond_1-2}
\end{equation}

As far as boundary conditions are concerned, we know from the general discussion set forth in Section~\ref{sect:basic} that we have to specify the distribution of incoming cars $\{f^1_{0j}\}_{j=1}^{n}$ along with the flux limiter $\Phi^1_{0,1}$ at the beginning of road $r=1$ and the flux limiters $\Phi^2_{1,2}$, $\Phi^3_{1,2}$ at the end of roads $r=2,\,3$. We set therefore:
\begin{equation}
	\begin{cases}
		f^1_{0j}=\bar{f}^1_j, & j=1,\,\dots,\,n \\
		\Phi^1_{0,1}=\bar{\Phi}^1 \\
		\Phi^r_{1,2}=\bar{\Phi}^r, & r=2,\,3,
	\end{cases}
\label{eq:bc_1-2}
\end{equation}
where the bar denotes prescribed (known) quantities. All the notations are summarized in Fig.~\ref{fig:1-2_junction-analysis} for quick reference.

Plugging Eqs.~\eqref{eq:cond_1-2},~\eqref{eq:bc_1-2} into Eq.~\eqref{eq:single.road.eq} gives the following problem for the 1-2 junction:
\begin{equation}
\begin{cases}
	\dfrac{df^1_{1j}}{dt}+v_j((a\Phi^2_{0,1}+(1-a)\Phi^3_{0,1})f^1_{1j}-\bar{\Phi}^1\bar{f}^1_j)
		=G_j[\f^1_1,\,\f^1_1]-f^1_{1j}L[\f^1_1] \\[3mm]
	\dfrac{df^2_{1j}}{dt}+v_j(\bar{\Phi}^2 f^2_{1j}-a\Phi^2_{0,1}f^1_{1j}\deltac)
		=G_j[\f^2_1,\,\f^2_1]-f^2_{1j}L[\f^2_1] \\[3mm]
	\dfrac{df^3_{1j}}{dt}+v_j(\bar{\Phi}^3 f^3_{1j}-(1-a)\Phi^3_{0,1}f^1_{1j}\deltac)
		=G_j[\f^3_1,\,\f^3_1]-f^3_{1j}L[\f^3_1] \\[3mm]
	f^r_{1j}(0)=\varphi^r_j \quad (r=1,\,2,\,3),
\end{cases}
\label{eq:IVP_1-2}
\end{equation}
where the gain and loss operators $G_j$, $L$ are defined in Eq.~\eqref{eq:GL} while the $\varphi^r_j$'s are the initial conditions, namely the prescribed kinetic distributions over all velocity classes for each road at time $t=0$. For analytical purposes it is convenient to rewrite Eq.~\eqref{eq:IVP_1-2} in mild form by formally integrating in time:
\begin{equation}
	\left\{
		\begin{aligned}
			f^1_{1j}(t)=\varphi^1_j+\ds{\int_0^t} & \left\{v_j(\bar{\Phi}^1(s)\bar{f}^1_j(s)-(a\Phi^2_{0,1}(s)+(1-a)\Phi^3_{0,1}(s))f^1_{1j}(s))\right. \\
				& \quad\left.+G_j[\f^1_1,\,\f^1_1](s)-f^1_{1j}(s)L[\f^1_1](s)\right\}\,ds \\
			f^2_{1j}(t)=\varphi^2_j+\ds{\int_0^t} & \left\{v_j(a\Phi^2_{0,1}(s)f^1_{1j}(s)\deltac-\bar{\Phi}^2(s)f^2_{1j}(s))\right. \\
				& \quad\left.+G_j[\f^2_1,\,\f^2_1](s)-f^2_{1j}(s)L[\f^2_1](s)\right\}\,ds \\
			f^3_{1j}(t)=\varphi^3_j+\ds{\int_0^t} & \left\{(v_j((1-a)\Phi^3_{0,1}(s)f^1_{1j}(s)\deltac-\bar{\Phi}^3(s)f^3_{1j}(s))\right. \\
				& \quad\left.+G_j[\f^3_1,\,\f^3_1](s)-f^3_{1j}(s)L[\f^3_1](s)\right\}\,ds.
		\end{aligned}
	\right.
	\label{eq:IVPmild_1-2}
\end{equation}

We are now in a position to define what we mean by kinetic solution across the 1-2 junction:
\begin{definition}
A distribution function
$$ \F_1(t)=(\f^1_1(t),\,\f^2_1(t),\,\f^3_1(t)):[0,\,T]\to\Rn\times\Rn\times\Rn, $$
where $\f^r_1=\{f^r_{1j}\}_{j=1}^{n}$, $r=1,\,2,\,3$, is said to be a \emph{mild kinetic solution across the 1-2 junction} if $\F_1\in\B^3$ and $\F_1$ satisfies Eq.~\eqref{eq:IVPmild_1-2}.
\label{def:mild.solution_1-2}
\end{definition}

Proving the existence and uniqueness of such a solution, along with its continuous dependence on the data, requires some technical assumptions which we summarize in the following:

\begin{framed}
\centering\textbf{Assumptions for Problem~\eqref{eq:IVP_1-2}}
\begin{enumerate}
\item \label{ip1} Initial and boundary data are such that $\bvarphi^r,\,\bar{\f}^1\in\B$, $r=1,\,2,\,3$.
\item \label{ip2} The function $\Phi:[0,\,1]^2\to [0,\,1]$ defining the flux limiters $\Phi^r_{0,1}=\Phi(\rho^r_0,\,\rho^r_1)$ for $r=2,\,3$ is such that:
\begin{enumerate}
\item[(i)] $0\leq\Phi(u,\,v)\leq 1$, $\forall\,u,\,v\in [0,\,1]$;
\item[(ii)] $\Phi(u,\,v)u\leq 1-v$, $\forall\,(u,\,v)\in [0,\,1]^2$ such that $u+v>1$;
\item[(iii)] $\exists\Lip(\Phi)>0$ such that 
	$$ \abs{\Phi(u_2,\,v_2)-\Phi(u_1,\,v_1)}\leq\Lip(\Phi)(\abs{u_2-u_1}+\abs{v_2-v_1}) $$
$\forall\,(u_1,\,v_1),\,(u_2,v_2)\in [0,\,1]^2$.
\end{enumerate}
\item \label{ip3} The flux limiter $\bar{\Phi}^1$ (boundary datum) at the entrance of the incoming road $r=1$ is like in Assumption~\ref{ip2}. In particular, $\bar{\Phi}^1=\Phi(\bar{\rho}^1,\,\rho^1_1)$ where $\bar{\rho}^1=\sum_{j=1}^{n}\bar{f}^1_{j}$ is the density of cars accessing road $r=1$ according to the boundary datum $\bar{\f}^1$.
\item \label{ip4} The elements $A^j_{hk}$ of the table of games satisfy property~\eqref{eq:tog.sum.1}. In addition, $\exists\Lip(A^j_{hk})>0$ such that
$$ \abs{A^j_{hk}[u]-A^j_{hk}[v]}\leq \Lip(A^j_{hk})\abs{u-v}, \quad \forall\,u,\,v\in[0,\,1]. $$
\end{enumerate}	
\end{framed}

\subsubsection{Uniqueness and continuous dependence}
We begin our analysis of the well-posedness of the problem of the 1-2 junction by an \emph{a priori} estimate, which entails the continuous dependence of the solution on the data and, consequently, its uniqueness for a given set of data. The implicit assumption is that (mild) solutions do actually exist, which will be proved later.
 
\begin{theorem}[Uniqueness and continuous dependence for the 1-2 junction]
Let $\{\{\bvarphi^r\}_{r=1}^{3},\,\bar{\f}^1,\,\{\bar{\Phi}^r_f\}_{r=1}^{3}\}$, $\{\{\bphi^r\}_{r=1}^{3},\,\bar{\g}^1,\,\{\bar{\Phi}^r_g\}_{r=1}^3\}$ be two sets of initial and boundary data for the 1-2 junction and $\F_1,\,\G_1\in\B^3$ two corresponding mild solutions of Problem~\eqref{eq:IVP_1-2}. Then there exists $\C>0$ such that
\begin{multline}
	\norm{\G_1-\F_1}_\infty\leq\C\left[\sum_{r=1}^{3}\norm{\bphi^r-\bvarphi^r}_1\right. \\
		\left.+\int_0^T\left(\sum_{r=1}^{3}\abs{\bar{\Phi}^r_g(t)-\bar{\Phi}^r_f(t)}+\norm{\bar{\g}^1(t)-\bar{\f}^1(t)}_1\right)dt\right].
	\label{eq:apriori_1-2}
\end{multline}
In particular, there is at most one solution corresponding to a given set of initial and boundary data.
\label{theo:uniqueness_1-2}
\end{theorem}
\begin{proof}
Subtracting term by term the mild equations~\eqref{eq:IVPmild_1-2} satisfied by $\F_1$ and $\G_1$, taking into account that $a,\,(1-a),\,v_j,\deltac\leq 1$, and summing over $j=1,\,\dots,\,n$, we have
\begin{equation}
	\left\{
		\begin{aligned}
			\norm{\g^1_1(t)-\f^1_1(t)}_1\leq\norm{\bphi^1-\bvarphi^1}_1 & +\sum_{j=1}^n\ds{\int_0^t}\abs{\bar{\Phi}^1_g(s)\bar{g}^1_j(s)-\bar{\Phi}^1_f(s)\bar{f}^1_j(s)}\,ds \\ 
				& \quad+\scrF^2(t)+\scrF^3(t)+\scrG^1(t)+\scrL^1(t) \\
			\norm{\g^2_1(t)-\f^2_1(t)}_1\leq\norm{\bphi^2-\bvarphi^2}_1 & +\sum_{j=1}^n\ds{\int_0^t}\abs{\bar{\Phi}^2_f(s)f^2_{1j}(s)-\bar{\Phi}^2_g(s)g^2_{1j}(s)}\,ds \\
				& \quad+\scrF^2(t)+\scrG^2(t)+\scrL^2(t) \\
			\norm{\g^3_1(t)-\f^3_1(t)}_1\leq\norm{\bphi^3-\bvarphi^3}_1 & +\sum_{j=1}^n\ds{\int_0^t}\abs{\bar{\Phi}^3_f(s)f^3_{1j}(s)-\bar{\Phi}^3_g(s)g^3_{1j}(s)}\,ds \\
				& \quad+\scrF^3(t)+\scrG^3(t)+\scrL^3(t),
		\end{aligned}
	\right.
	\label{eq:FGL_1-2}
\end{equation}
where we have set
\begin{align*}
	\scrF^r(t) &:= \sum_{j=1}^n\int_0^t\abs{\Phi(\rho^r_0,\,\rho^r_1)(s)f^1_{1j}(s)-\Phi(\varrho^r_0,\,\varrho^r_1)(s)g^1_{1j}(s)}\,ds, \\
	\scrG^r(t) &:= \sum_{j=1}^n\int_0^t\abs{G_j[\g^r_1,\,\g^r_1](s)-G_j[\f^r_1,\,\f^r_1](s)}\,ds, \\
	\scrL^r(t) &:= \sum_{j=1}^n\int_0^t\abs{f^r_{1j}(s)L[\f^r_1](s)-g^r_{1j}(s)L[\g^r_1](s)}\,ds
\end{align*}
for $r=1,\,2,\,3$. Notice that in the term $\scrF^r$ we have occasionally written explicitly the dependence of the flux limiter on the densities for the sake of the next estimates. In particular, we agree that the symbol $\rho^r_i$ stands for the density on the $r$th road computed with respect to the distribution function $\f^r_i$ while $\varrho^r_i$ stands for that computed with respect to $\g^r_i$.

Now, using Assumptions~\ref{ip1}--\ref{ip3}, by standard calculations we get the following estimates for the second terms at the right-hand side of Eqs.~\eqref{eq:FGL_1-2}:
\begin{align*}
	\abs{\bar{\Phi}^1_g(s)\bar{g}^1_j(s)-\bar{\Phi}^1_f(s)\bar{f}^1_j(s)} &\leq \abs{\bar{\Phi}^1_g(s)-\bar{\Phi}^1_f(s)}+\abs{\bar{g}^1_j(s)-\bar{f}^1_j(s)} \\
	\abs{\bar{\Phi}^r_f(s)f^r_{1j}(s)-\bar{\Phi}^r_g(s)g^r_{1j}(s)} &\leq \abs{\bar{\Phi}^r_f(s)-\bar{\Phi}^r_g(s)}+\abs{f^r_{1j}(s)-g^r_{1j}(s)}.
\end{align*}
Moreover, owing to Assumptions~\ref{ip1},~\ref{ip2}, we also discover for the $\scrF^r$'s:
\begin{align*}
	& \abs{\Phi(\rho^r_0,\,\rho^r_1)(s)f^1_{1j}(s)-\Phi(\varrho^r_0,\,\varrho^r_1)(s)g^1_{1j}(s)} \\
		&\qquad\qquad\qquad \leq\abs{\Phi(\rho^r_0,\,\rho^r_1)(s)-\Phi(\varrho^r_0,\,\varrho^r_1)(s)}+\abs{f^1_{1j}(s)-g^1_{1j}(s)} \\ 
		&\qquad\qquad\qquad =\abs{\Phi(\rho^r_0,\,\rho^r_1)-\Phi(\varrho^r_0,\,\varrho^r_1)}+\abs{f^1_{1j}(s)-g^1_{1j}(s)} \\ 
		&\qquad\qquad\qquad \leq\C\left(\abs{\rho^r_0-\varrho^r_0}+\abs{\rho^r_1-\varrho^r_1}\right)+\abs{f^1_{1j}(s)-g^1_{1j}(s)} \\
	\intertext{whence, using the transmission conditions~\eqref{eq:cond_1-2} for handling $\rho^r_0$ and $\varrho^r_0$,}
		&\qquad\qquad\qquad \leq\C(\norm{\g^r_1(s)-\f^r_1(s)}_1+\norm{\g^1_1(s)-\f^1_1(s)}_1)
\end{align*}
where $\C$ is a positive constant whose specific value is unimportant (it can even change from line to line). Furthermore, concerning the terms $\scrG^r$, $\scrL^r$, it results (cf. the proof of Theorem 5.2 in~\cite{fermo2013SIAP}):
\begin{equation*}
	\left.
	\begin{array}{r}
		\ds{\sum_{j=1}^n}\abs{G_j[\g^r_1,\,\g^r_1](s)-G_j[\f^r_1,\,\f^r_1](s)} \\
		\ds{\sum_{j=1}^n}\abs{g^r_{1j}L[\g^r_1](s)-f^r_{1j}L[\f^r_1](s)}
	\end{array}
	\right\}
	\leq\C\norm{\g^r(s)-\f^r_1(s)}_1.
\end{equation*}

Finally, collecting all of the estimates obtained so far and summing term by term Eqs.~\eqref{eq:FGL_1-2} we arrive at
\begin{multline*}
	\norm{\G_1(t)-\F_1(t)}_1\leq\sum_{r=1}^{3}\norm{\bphi^r-\bvarphi^r}_1 \\
		+\C\int_0^t\left(\norm{\G_1(s)-\F_1(s)}_1+\sum_{r=1}^{3}\abs{\bar{\Phi}^r_g(s)-\bar{\Phi}^r_f(s)}+\norm{\bar{\g}^1(s)-\bar{\f}^1(s)}_1\right)\,ds
\end{multline*}
whence, owing to Gronwall's inequality,
\begin{multline*}
	\norm{\G_1(t)-\F_1(t)}_1\leq e^{\C t}\left[\sum_{r=1}^{3}\norm{\bphi^r-\bvarphi^r}_1\right. \\
		\left. +\int_0^t\left(\sum_{r=1}^{3}\abs{\bar{\Phi}^r_g(s)-\bar{\Phi}^r_f(s)}+\norm{\bar{\g}^1(s)-\bar{\f}^1(s)}_1\right)\,ds\right].
\end{multline*}

Taking the supremum over $t\in [0,\,T]$ of both sides we obtain the first assertion of the theorem, namely the continuous dependence estimate. Uniqueness of the solution follows straightforwardly by taking $\bvarphi^r=\bphi^r$, $\bar{\Phi}^r_f=\bar{\Phi}^r_g$ for $r=1,\,2,\,3$, and $\bar{\f}^1=\bar{\g}^1$.
\end{proof}

\subsubsection{Existence}
The next theorem states that the unique mild kinetic solution across the 1-2 junction does indeed exist in $\B^3$.

\begin{theorem}[Existence for the 1-2 junction]
There exists $\F_1\in\B^3$ which is a mild solution to Problem~\eqref{eq:IVP_1-2} in the sense of Definition~\ref{def:mild.solution_1-2}.
\label{theo:existence_1-2}
\end{theorem}
\begin{proof}
The proof is organized in three steps according to the same guidelines as the existence theorem proved in~\cite{fermo2013SIAP}. However, for the reader's convenience, we repropose the relevant material from~\cite{fermo2013SIAP}, proving especially the different arguments, so as to make our exposition self-contained.

\paragraph*{Step 1: Discrete-in-time model}
We first consider the model at discrete time instants $t^\kk_\n=\n\Delta{t}_\kk$, where the index $\n$ labels the discrete time, the index $\kk$ is a mesh parameter denoting the level of refinement of the time grid, and the time step $\Delta{t}_\kk$ is chosen such that it tends to zero for $\kk\to\infty$:
\begin{equation}
	\left\{
		\begin{aligned}[c]
			f^{1,\n+1,\kk}_{1j} &= f^{1,\n,\kk}_{1j}-\Delta{t}_\kk v_j((a\Phi^{2,\n,\kk}_{0,1}+(1-a)\Phi^{3,\n,\kk}_{0,1})f^{1,\n,\kk}_{1j}-
				\bar{\Phi}^{1,\n,\kk} \bar{f}^{1,\n,\kk}_j) \\
					&\phantom{=} +\Delta{t}_\kk(G_j[\f^{1,\n,\kk}_1,\,\f^{1,\n,\kk}_1]-f^{1,\n,\kk}_{1j}L[\f^{1,\n,\kk}_1]) \\[2mm]
			f^{2,\n+1,\kk}_{1j} &= f^{2,\n,\kk}_{1j}-\Delta{t}_\kk v_j(\bar{\Phi}^{2,\n,\kk}f^{2,\n,\kk}_{1j}-a\Phi^{2,\n,\kk}_{0,1}f^{1,\n,\kk}_{1j}\deltac) \\
				&\phantom{=} +\Delta{t}_\kk(G_j[\f^{2,\n,\kk}_1,\,\f^{2,\n,\kk}_1]-f^{2,\n,\kk}_{1j}L[\f^{2,\n,\kk}_1]) \\[2mm]
			f^{3,\n+1,\kk}_{1j} &= f^{3,\n,\kk}_{1j}-\Delta{t}_\kk v_j(\bar{\Phi}^{3,\n,\kk} f^{3,\n,\kk}_{1j}-(1-a)\Phi^{3,\n,\kk}_{0,1}f^{1,\n,\kk}_{1j}\deltac) \\
				&\phantom{=} +\Delta{t}_\kk(G_j[\f^{3,\n,\kk}_1,\,\f^{3,\n,\kk}_1]-f^{3,\n,\kk}_{1j}L[\f^{3,\n,\kk}_1]).
		\end{aligned}
	\right.
	\label{eq:discrete.time_1-2}
\end{equation}

We claim that if $\Delta{t}_\kk$ is sufficiently small the iterates $\f^{r,\n,\kk}_1$ (understood as constant functions of $t$) belong to $\B$ for all $\n,\,\kk\geq 0$ and all $r=1,\,2,\,3$. This can be proved by induction from Eq.~\eqref{eq:discrete.time_1-2}, assuming that $\f^{r,\n,\kk}_1\in\B$ and checking that $\f^{r,\n+1,\kk}_1\in\B$ as well. More specifically:
\begin{enumerate}
\item \emph{Non-negativity of the iterates.} Taking into account the boundedness between $0$ and $1$ of the flux limiters, the boundary data, the speed classes, the distribution function $\f^{r,\n,\kk}_1$ itself (by inductive assumption), and the non-negativity of $G_j[\f^{r,\n,\kk}_1,\,\f^{r,\n,\kk}_1]$ we easily see from Eq.~\eqref{eq:discrete.time_1-2} that
$$ f^{r,\n+1,\kk}_{1j}\geq f^{r,\n,\kk}_{1j}-\Delta{t}_\kk(1+f^{r,\n,\kk}_{1j}L[\f^{r,\n,\kk}_1]), \quad \forall\,r=1,\,2,\,3. $$
Moreover, since $L[\f^{r,\n,\kk}_1]=\eta_0{(\rho^{r,\n,\kk}_1)}^2\leq\eta_0$ (because the inductive assumption $\f^{r,\n,\kk}_1\in\B$ implies in particular $\rho^{r,\n,\kk}_1=\sum_{j=1}^{n}f^{r,\n,\kk}_{1j}\leq 1$) we deduce
$$ f^{r,\n+1,\kk}_{1j}\geq f^{r,\n,\kk}_{1j}(1-\Delta{t}_\kk(1+\eta_0)), \quad \forall\,r=1,\,2,\,3 $$
whence if $\Delta{t}_\kk<1/(1+\eta_0)$ we conclude that $f^{r,\n+1,\kk}_{1j}\geq 0$.

\item \emph{Boundedness of the iterates.} Since $v_j\leq 1$ for all $j$, we have
\begin{equation}
	\left\{
	\begin{aligned}[c]
		f^{1,\n+1,\kk}_{1j} &\leq f^{1,\n,\kk}_{1j}+\Delta{t}_\kk\left(\bar{\Phi}^{1,\n,\kk}\bar{f}^{1,\n,\kk}_j\right. \\
			&\phantom{\leq} \left.+G_j[\f^{1,\n,\kk}_1,\,\f^{1,\n,\kk}_1]-f^{1,\n,\kk}_{1j}L[\f^{1,\n,\kk}]\right) \\[2mm]
		f^{2,\n+1,\kk}_{1j} &\leq f^{2,\n,\kk}_{1j}+\Delta{t}_\kk\left(a\Phi^{2,\n,\kk}_{0,1}f^{1,\n,\kk}_{1j}\deltac\right. \\
			&\phantom{\leq} \left.+G_j[\f^{2,\n,\kk}_1,\,\f^{2,\n,\kk}_1]-f^{2,\n,\kk}_{1j}L[\f^{2,\n,\kk}_1]\right) \\[2mm]
		f^{3,\n+1,\kk}_{1j} &\leq f^{3,\n,\kk}_{1j}+\Delta{t}_\kk\left((1-a)\Phi^{3,\n,\kk}_{0,1}f^{1,\n,\kk}_{1j}\deltac\right. \\
			&\phantom{\leq} \left.+G_j[\f^{3,\n,\kk}_1,\,\f^{3,\n,\kk}_1]-f^{3,\n,\kk}_{1j}L[\f^{3,\n,\kk}_1]\right).
	\end{aligned}
	\right.
	\label{eq:boundedness}
\end{equation}
But
$$ G_j[\f^{r,\n,\kk}_1,\,\f^{r,\n,\kk}_1]-f^{r,\n,\kk}_{1j}L[\f^{r,\n,\kk}_1]\leq\eta_0\left[1-{(f^{r,\n,\kk}_{1j})}^2\right] $$
whereas owing to Assumptions~\ref{ip2},~\ref{ip3} we can write
\begin{align}
	\begin{aligned}[b]
	\bar{\Phi}^{1,\n,\kk}\bar{f}^{1,\n,\kk}_j &= \Phi(\bar{\rho}^{1,\n,\kk},\,\rho^{1,\n,\kk}_1)\bar{f}^{1,\n,\kk}_j
		\leq\Phi(\bar{\rho}^{1,\n,\kk},\,\rho^{1,\n,\kk}_1)\bar{\rho}^{1,\n,\kk} \\
	& \leq 1-\rho^{1,\n,\kk}_1\leq 1-f^{1,\n,\kk}_{1j},
	\end{aligned}
	\label{eq:bound1}
\end{align}
and likewise, recalling also the transmission conditions~\eqref{eq:cond_1-2},
\begin{align}
	\begin{aligned}[b]
	a\Phi^{2,\n,\kk}_{0,1}f^{1,\n,\kk}_{1j}\deltac &\leq a\Phi\left(a\sum_{j=2}^{n}f^{1,\n,\kk}_{1j},\,\rho^{2,\n,\kk}_1\right)\sum_{j=2}^{n}f^{1,\n,\kk}_{1j} \\
		&\leq 1-\rho^{2,\n,\kk}_1\leq 1-f^{2,\n,\kk}_{1j}
	\end{aligned}
	\label{eq:bound2-1}
\end{align}
and
\begin{align}
	\begin{aligned}[b]
	(1-a)\Phi^{3,\n,\kk}_{0,1}f^{1,\n,\kk}_{1j}\deltac &\leq (1-a)\Phi\left((1-a)\sum_{j=2}^{n}f^{1,\n,\kk}_{1j},\,\rho^{3,\n,\kk}_1\right)\sum_{j=2}^{n}f^{1,\n,\kk}_{1j} \\
		&\leq 1-\rho^{3,\n,\kk}_1\leq 1-f^{3,\n,\kk}_{1j}.
	\end{aligned}
	\label{eq:bound2-2}
\end{align}
Consequently, for all $r=1,\,2,\,3$ it results
\begin{align*}
	f^{r,\n+1,\kk}_{1j} &\leq f^{r,\n,\kk}_{1j}+\Delta{t}_\kk\left((1-f^{r,\n,\kk}_{1j})+\eta_0(1-{(f^{r,\n,\kk}_{1j})}^2\right) \\
		&\leq f^{r,\n,\kk}_{1j}+\Delta{t}_\kk\left((1-f^{r,\n,\kk}_{1j})+2\eta_0(1-f^{r,\n,\kk}_{1j})\right) \\
		&= f^{r,\n,\kk}_{1j}+\Delta{t}_\kk(1+2\eta_0)(1-f^{r,\n,\kk}_{1j}),
\end{align*}
thus choosing $\Delta{t}_\kk<1/(1+2\eta_0)$ yields $f^{r,\n+1,\kk}_{1j}\leq 1$.

\item \emph{Boundedness of the sum of the iterates.} By Eq.~\eqref{eq:boundedness} and taking furthermore Eqs.~\eqref{eq:conservativeness},~\eqref{eq:bound1},~\eqref{eq:bound2-1},~\eqref{eq:bound2-2} into account we have
$$ \sum_{j=1}^{n}f^{r,\n+1,\kk}_{1j}\leq\rho^{r,\n,\kk}_1+\Delta{t}_\kk(1-\rho^{r,\n,\kk}_1)\leq 1 $$
provided $\Delta{t}_\kk \leq 1$.
\end{enumerate}
Finally, the three properties hold simultaneously if $\Delta{t}_\kk<1/(1+2\eta_0)$.

\paragraph*{Step 2: From discrete to continuous time} 
Now we pass from discrete to continuous time. To this end, we introduce the function $\hat{\f}^{r,\kk}_1:[0,\,T]\to\R^{n}$ interpolating piecewise linearly in time the iterates $\{\f^{r,\n,\kk}_1\}_{\n=1}^{N_\kk}$:
\begin{equation}
	\hat{\f}^{r,\kk}_1(t)=\sum_{\n=1}^{N_\kk}\left[\left(1-\frac{t-t^\kk_{\n-1}}{\Delta{t}_\kk}\right)\f^{r,\n-1,\kk}_1
		+\frac{t-t^\kk_{\n-1}}{\Delta{t}_\kk}\f^{r,\n,\kk}_1\right]\ind_{[t^\kk_{\n-1},\,t^\kk_\n]}(t),
	\label{eq:interpolante}
\end{equation}
where $N_\kk$ is the total number of time steps. We assume that $N_\kk$ and $\Delta{t}_\kk$ are chosen in such a way that $N_\kk\Delta{t}_\kk=T$ independently of the refinement parameter $\kk$. Then an analysis similar to that performed in the proof of Lemma 5.4 in~\cite{fermo2013SIAP}, which uses the properties of the $\f^{r,\n,\kk}_1$'s proved in the previous Step 1 along with Ascoli-Arzel\`{a} compactness criterion in $X_T$, shows that, in the limit $\kk\to\infty$, the interpolation~\eqref{eq:interpolante} converges to a distribution function $\f^r_1\in\B$, i.e.,
\begin{equation}
	\lim_{\kk\to\infty} \norm{\hat{\f}^{r,\kk}_1-\f^r_1}_\infty=0
	\label{eq:limit_k}
\end{equation}
for all $r=1,\,2,\,3$.

\paragraph*{Step 3: Construction of the solution}
Finally, we prove that the function $\F_1=(\f^1_1,\,\f^2_1,\,\f^3_1)$, where the $\f^r_1$'s are those appearing in Eq.~\eqref{eq:limit_k}, is a mild solution to Problem~\eqref{eq:IVP_1-2}. Owing to Eq.~\eqref{eq:limit_k}, we can expect the $\hat{\f}^{r,\kk}_1$'s defined by Eq.~\eqref{eq:interpolante} to be an approximation of $\f^r_1$ for every fixed $\kk$. Therefore, plugging the components $\hat{f}^{r,\kk}_{1j}$ into Eq.~\eqref{eq:IVPmild_1-2} produces some reminders at the right-hand side that we write as follows:
\begin{equation}
	\left\{
	\begin{aligned}[c]
	\hat{f}^{1,\kk}_{1j}(t)-\varphi^1_j+\int_0^t &\left\{v_j\left((a\Phi^{2,\kk}_{0,1}(s)+(1-a)\Phi^{3,\kk}_{0,1}(s))\hat{f}^{1,\kk}_{1j}(s)
		-\bar{\Phi}^{1,\kk}(s)\bar{f}^{1,\kk}_j(s))\right)\right. \\ 
		&\left. -G_j[\hat{\f}^{1,\kk}_1,\,\hat{\f}^{1,\kk}_1](s)+\hat{f}^{1,\kk}_{1j}(s)L[\hat{\f}^{1,\kk}_1](s)\right\}\,ds=\int_0^t e^{1,\kk}_{1j}(s)\,ds \\
	\hat{f}^{2,\kk}_{1j}(t)-\varphi^2_j+\int_0^t &\left\{v_j\left(\bar{\Phi}^{2,\kk}(s)\hat{f}^{2,\kk}_{1j}(s)-a\Phi^{2,\kk}_{0,1}(s)\hat{f}^{1,\kk}_{1j}(s)\deltac\right)\right. \\
		& \left. -G_j[\hat{\f}^{2,\kk}_1,\,\hat{\f}^{2,\kk}_1](s)+\hat{f}^{2,\kk}_{1j}(s)L[\hat{\f}^{2,\kk}_1](s)\right\}\,ds=\int_0^t e^{2,\kk}_{1j}(s)\,ds \\
	\hat{f}^{3,\kk}_{1j}(t)-\varphi^3_j+\int_0^t &\left\{v_j\left(\bar{\Phi}^{3,\kk}(s)\hat{f}^{3,\kk}_{1j}(s)-(1-a)\Phi^{3,\kk}_{0,1}(s)\hat{f}^{1,\kk}_{1j}(s)\deltac\right)\right. \\
		& \left. -G_j[\hat{\f}^{3,\kk}_1,\,\hat{\f}^{3,\kk}_1](s)+\hat{f}^{3,\kk}_{1j}(s)L[\hat{\f}^{3,\kk}_1](s)\right\}\,ds=\int_0^t e^{3,\kk}_{1j}(s)\,ds
	\end{aligned}
	\label{eq:int_error}
	\right.
\end{equation}
where we have denoted $\Phi^{2,\kk}_{0,1}:=\Phi(\hat{\rho}^{2,\kk}_0,\,\hat{\rho}^{2,\kk}_1)$ and $\Phi^{3,\kk}_{0,1}:=\Phi(\hat{\rho}^{3,\kk}_0,\,\hat{\rho}^{3,\kk}_1)$.

The idea is now to take the limit $\kk\to\infty$ in Eq.~\eqref{eq:int_error} by exploiting Eq.~\eqref{eq:limit_k}, showing that the left-hand sides converge to the corresponding expressions evaluated for $f^r_{1j}(t)$ while the right-hand sides go to zero.

First, let us consider the left-hand sides of Eq.~\eqref{eq:int_error}. Note that all terms appearing in the integrals are bounded from above by an integrable constant, thus by dominated convergence it is possible to commute the limit in $\kk$ with the integral in $t$. Next, using Eq.~\eqref{eq:limit_k} and taking Eq.~\eqref{eq:apriori_1-2} into account with $\G_1=\hat{\F}^\kk_1$ we conclude that the left-hand sides of Eq.~\eqref{eq:int_error} converge to the analogous expression with $\F_1$.

At this point it remains to prove that the right-hand sides of Eq.~\eqref{eq:int_error} tend to zero. To this end, we take the time derivative of both sides using the fact that the $\hat{\f}^{r,\kk}_1$'s are Lipschitz continuous functions of $t$ (cf. Eq.~\eqref{eq:interpolante}), hence almost everywhere differentiable owing to Rademacher's Theorem:
\begin{equation}
	\left\{
	\begin{aligned}[c]
		\dfrac{d\hat{f}^1_{1j}}{dt}+v_j &\left((a\Phi^{2,\kk}_{0,1}(t)+(1-a)\Phi^{3,\kk}_{0,1}(t))\hat{f}^{1,\kk}_{1j}(t)
			-\bar{\Phi}^1\bar{f}^{1,\kk}_j(t)\right) \\
				&-G_j[\hat{\f}^{1,\kk}_1,\,\hat{\f}^{1,\kk}_1](t)+\hat{f}^{1,\kk}_{1j}(t)L[\hat{\f}^{1,\kk}_1](t)=e^{1,\kk}_{1j}(t) \\[2mm]
		\dfrac{d\hat{f}^2_{1j}}{dt}+v_j &\left(\bar{\Phi}^2(t)\hat{f}^{2,\kk}_{1j}(t)-a\Phi^{2,\kk}_{0,1}(t)\hat{f}^{1,\kk}_{1j}(t)\deltac\right) \\
			& -G_j[\hat{\f}^{2,\kk}_1,\,\hat{\f}^{2,\kk}_1](t)+\hat{f}^{2,\kk}_{1j}(t)L[\hat{\f}^{2,\kk}_1](t)=e^{2,\kk}_{1j}(t) \\[2mm]
		\dfrac{d\hat{f}^3_{1j}}{dt}+v_j &\left(\bar{\Phi}^3(t)\hat{f}^{3,\kk}_{1j}(t)-(1-a)\Phi^{3,\kk}_{0,1}(t)\hat{f}^{1,\kk}_{1j}(t)\deltac\right) \\
			& -G_j[\hat{\f}^{3,\kk}_1,\,\hat{\f}^{3,\kk}_1](t)+\hat{f}^{3,\kk}_{1j}(t)L[\hat{\f}^{3,\kk}_1](t)=e^{3,\kk}_{1j}(t).
	\end{aligned}
	\right.
	\label{eq:IVP-interpolante}
\end{equation}

From Eq.~\eqref{eq:interpolante} it results
$$ \dfrac{d\hat{f}^{r,\kk}_{1j}}{dt}=\dfrac{1}{\Delta t_\kk}\sum_{\n=1}^{N_\kk}\left(f^{r,\n,\kk}_{1j}-f^{r,\n-1,\kk}_{1j}\right)\ind_{[t^\kk_{\n-1},\,t^\kk_\n]}(t) $$
whence, using Eq.~\eqref{eq:discrete.time_1-2} to manipulate the differences in brackets, we get:
\begin{itemize}
\item $
	\begin{aligned}[t]
		\dfrac{d\hat{f}^{1,\kk}_{1j}}{dt} &= \ds{\sum_{\n=1}^{N_\kk}}\left\{-v_j\left((a\Phi^{2,\n-1,\kk}_{0,1}+(1-a)\Phi^{3,\n-1,\kk}_{0,1})f^{1,\n-1,\kk}_{1j}-
			\bar{\Phi}^1\bar{f}^{1,\n-1,\kk}_j\right)\right. \\
		&\phantom{=} \left.+G_j[\f^{1,\n-1,\kk}_1,\,\f^{1,\n-1,\kk}_1]-f^{1,\n-1,\kk}_{1j}L[\f^{1,\n-1,\kk}_1]\right\}\ind_{[t^\kk_{\n-1},\,t^\kk_\n]}(t)
	\end{aligned}
	$

\item $
	\begin{aligned}[t]
		\dfrac{d\hat{f}^{2,\kk}_{1j}}{dt} &= \ds{\sum_{\n=1}^{N_\kk}}\left\{v_j\left(\bar{\Phi}^2 f^{2,\n-1,\kk}_{1j}-a\Phi^{2,\n-1,\kk}_{0,1}f^{1,\n-1,\kk}_{1j}\deltac\right)\right. \\
			&\phantom{=} \left.+G_j[\f^{2,\n-1,\kk}_1,\,\f^{2,\n-1,\kk}_1]-f^{2,\n-1,\kk}_{1j}L[\f^{2,\n-1,\kk}_1]\right\}\ind_{[t^\kk_{\n-1},\,t^\kk_\n]}(t)
	\end{aligned}
	$

\item $
	\begin{aligned}[t]
		\dfrac{d\hat{f}^{3,\kk}_{1j}}{dt} &= \ds{\sum_{\n=1}^{N_\kk}}\left\{v_j\left(\bar{\Phi}^3 f^{3,\n-1,\kk}_{1j}-(1-a)\Phi^{3,\n-1,\kk}_{0,1}f^{1,\n-1,\kk}_{1j}\deltac\right)\right. \\
			&\phantom{=} \left.+G_j[\f^{3,\n-1,\kk}_1,\,\f^{3,\n-1,\kk}_1]-f^{3,\n-1,\kk}_{1j}L[\f^{3,\n-1,\kk}_1]\right\}\ind_{[t^\kk_{\n-1},\,t^\kk_\n]}(t)
	\end{aligned}
	$
\end{itemize}
where we have denoted $\Phi^{r,\n-1,\kk}_{0,1}:=\Phi(\rho^{r,\n-1,\kk}_0,\,\rho^{r,\n-1,\kk}_1)$.

Furthermore, continuing from Eq.~\eqref{eq:IVP-interpolante} and developing the other terms according to Eq.~\eqref{eq:interpolante} we discover:
\begin{itemize}
\item $
	\begin{aligned}[t]
		v_j &\left(a\Phi^{2,\kk}_{0,1}+(1-a)\Phi^{3,\kk}_{0,1}\right)\hat{f}^{1,\kk}_{1j} \\
		& =v_j\ds{\sum_{\n=1}^{N_\kk}}\Bigl\{\left(a\Phi^{2,\kk}_{0,1}+(1-a)\Phi^{3,\kk}_{0,1}\right)f^{1,\n-1,\kk}_{1j} \\
		&\phantom{=} +\dfrac{t-t^\kk_{\n-1}}{\Delta{t}_\kk}\left(a\Phi^{2,\kk}_{0,1}+(1-a)\Phi^{3,\kk}_{0,1}\right)
			\left(f^{1,\n,\kk}_{1j}-f^{1,\n-1,\kk}_{1j}\right)\Bigr\}\ind_{[t^\kk_{\n-1},\,t^\kk_{\n}]}(t)
	\end{aligned}
	$
	
\item $
	\begin{aligned}[t]
		v_j &\left(\bar{\Phi}^{2,\kk}\hat{f}^{2,\kk}_{1j}-a\Phi^{2,\kk}_{1j}\hat{f}^{1,\kk}_{1j}\deltac\right) \\
		& =v_j\ds{\sum_{\n=1}^{N_\kk}}\Bigl\{\left(\bar{\Phi}^{2,\kk}f^{2,\n-1,\kk}_{1j}-a\Phi^{2,\kk}_{0,1}f^{1,\n-1,\kk}_{1j}\deltac\right) \\
		&\phantom{=} +\dfrac{t-t^\kk_{\n-1}}{\Delta{t}_\kk}\left[\bar{\Phi}^{2,\kk}\left(f^{2,\n,\kk}_{1j}-f^{2,\n-1,\kk}_{1j}\right)
			-a\Phi^{2,\kk}_{0,1}\left(f^{1,\n,\kk}_{1j}-f^{1,\n-1,\kk}_{1j}\right)\deltac\right]\Bigl\}\ind_{[t^\kk_{\n-1},\,t^\kk_{\n}]}(t)
	\end{aligned}
	$
	
\item $
	\begin{aligned}[t]
		v_j &\left(\bar{\Phi}^{3,\kk}\hat{f}^{3,\kk}_{1j}-(1-a)\Phi^{3,\kk}_{0,1}\hat{f}^{1,\kk}_{1j}\deltac\right) \\
		& =v_j\ds{\sum_{\n=1}^{N_\kk}}\Bigl\{\left[\bar{\Phi}^{3,\kk} f^{3,\n-1,\kk}_{1j}-(1-a)\Phi^{3,\kk}_{0,1}f^{1,\n-1,\kk}_{1j}\deltac\right] \\
		&\phantom{=} +\dfrac{t-t^\kk_{\n-1}}{\Delta{t}_\kk}\left[\bar{\Phi}^{3,\kk}\left(f^{3,\n,\kk}_{1j}-f^{3,\n-1,\kk}_{1j}\right)
			-(1-a)\Phi^{3,\kk}_{0,1}\left(f^{1,\n,\kk}_{1j}-f^{1,\n-1,\kk}_{1j}\right)\deltac\right]\Bigl\}\ind_{[t^\kk_{\n-1},\,t^\kk_{\n}]}(t)
	\end{aligned}
	$
	
\item $
	\begin{aligned}[t]
		G_j[\hat{\f}^{r,\kk}_1,\,\hat{\f}^{r,\kk}_1] &= \eta_0\hat{\rho}^{r,\kk}_1\sum_{h,k=1}^{n}\sum_{\n=1}^{N_\kk}
			A^j_{hk}[\hat{\rho}^{r,\kk}]\Bigl\{f^{r,\n-1,\kk}_{1h}f^{r,\n-1,\kk}_{1k} \\
		& +\dfrac{t-t^\kk_{\n-1}}{\Delta{t}_\kk}f^{r,\n-1,\kk}_{1h}\left(f^{r,\n,\kk}_{1k}-f^{r,\n-1,\kk}_{1k}\right) \\
		& +\dfrac{t-t^\kk_{\n-1}}{\Delta{t}_\kk}f^{r,\n-1,\kk}_{1k}\left(f^{r,\n,\kk}_{1h}-f^{r,\n-1,\kk}_{1h}\right) \\
		& +\left(\dfrac{t-t^\kk_{\n-1}}{\Delta{t}_\kk}\right)^2\left(f^{r,\n,\kk}_{1h}-f^{r,\n-1,\kk}_{1h}\right)
			\left(f^{r,\n,\kk}_{1k}-f^{r,\n-1,\kk}_{1k}\right)\Bigr\}\ind_{[t^\kk_{\n-1},\,t^\kk_{\n}]}(t)
	\end{aligned}
	$

\item $
	\begin{aligned}[t]
		\hat{f}^{r,\kk}_{1j}L[\hat{\f}^{r,\kk}_1] &= \eta_0\hat{\rho}^{r,\kk}_1\sum_{k=1}^{n}\sum_{\n=1}^{N_\kk}\Bigl\{f^{r,\n-1,\kk}_{1j}f^{r,\n-1,\kk}_{1k} \\
			& +\dfrac{t-t^\kk_{\n-1}}{\Delta{t}_\kk}f^{r,\n-1,\kk}_{1j}\left(f^{r,\n,\kk}_{1k}-f^{r,\n-1,\kk}_{1k}\right) \\
			& +\dfrac{t-t^\kk_{\n-1}}{\Delta{t}_\kk}f^{r,\n-1,\kk}_{1k}\left(f^{r,\n,\kk}_{1j}-f^{r,\n-1,\kk}_{1j}\right) \\
			& +\left(\dfrac{t-t^\kk_{\n-1}}{\Delta t_\kk}\right)^2\left(f^{r,\n,\kk}_{1j}-f^{r,\n-1,\kk}_{1j}\right)
				\left(f^{r,\n,\kk}_{1k}-f^{r,\n-1,\kk}_{1k}\right)\Bigr\}\ind_{[t^\kk_{\n-1},\,t^\kk_{\n}]}(t).
	\end{aligned}
	$
\end{itemize}

Inserting the quantities just computed at the left-hand side of Eq.~\eqref{eq:IVP-interpolante} we get an expression of $e^{r,\kk}_{1j}$ in which terms featuring the difference $f^{r,\n,\kk}_{1j}-f^{r,\n-1,\kk}_{1j}$ appear along with others involving differences between pairs of flux limiters and tables of games evaluated at $\f^{r,\n-1,\kk}$ and $\hat{\f}^{r,\kk}$, respectively. Applying the following estimate, proved in~\cite{fermo2013SIAP} (cf. Lemma 5.4),
$$ \abs{f^{r,\n,\kk}_{1j}-f^{r,\n-1,\kk}_{1j}}\leq 2(1+\eta_0)\Delta{t}_\kk $$
to the first terms and invoking the Lipschitz continuity of the flux limiters and the table of games (owing to Assumptions~\ref{ip2}--\ref{ip4}) we finally obtain:
$$ \abs{e^{r,\kk}_{1j}(t)}\leq\C\Delta{t}_\kk, \quad \forall\,r=1,\,2,\,3 $$
for a suitable constant $\C>0$ independent of $\kk$, whence
$$ \int_0^t\abs{e^{r,\kk}_{1j}(s)}\,ds\leq\C\Delta{t}_\kk T\xrightarrow{\kk\to\infty}0, \quad \forall\,r=1,\,2,\,3, $$
which completes the proof.
\end{proof}

\subsubsection{Regularity}
We conclude the analysis of the 1-2 junction by showing that the unique mild solution found thus far is actually a classical one.

\begin{corollary}[Improved regularity for the solution of the 1-2 junction]
The mild solution to Problem~\eqref{eq:IVP_1-2} is of class $C^1$ in $[0,\,T]$.
\end{corollary}
\begin{proof}
The continuity of $\F_1$ in Theorem~\ref{theo:existence_1-2} together with Assumptions~\ref{ip2},~\ref{ip4} imply that both the second term at the left-hand side and the terms at the right-hand side of Eqs.~\eqref{eq:IVP_1-2} are continuous. Thus $\frac{df^{r}_{1j}}{dt}$ are continuous for every $j=1,\,\dots,\,n$ and every $r=1,\,2,\,3$.
\end{proof}

\subsection{Well-posedness of the 2-1 junction}
We now consider the model of the 2-1 junction described in Section~\ref{sect:2-1_junction}, still assuming for simplicity that all roads have $m_r=1$ cells. Moreover, to fix the ideas and without loss of generality, we also assume that the road with right-of-way is $r^\ast=1$. Thus the transmission conditions~\eqref{eq:f0j_2-1},~\eqref{eq:Phi_2-1_1},~\eqref{eq:Phi_2-1_2} become:
\begin{itemize}
\item if $q^1_1+q^2_1\leq p$ then
$$ f^3_{0j}=(f^1_{1j}+f^2_{1j})\deltac, \qquad \Phi^1_{1,2}=\Phi^3_{0,1}, \qquad \Phi^2_{1,2}=\Phi^3_{0,1}, $$
\item if $q^1_{1}+q^2_{1}>p$ then
$$ f^3_{0j}=f^1_{1j}\deltac, \qquad \Phi^1_{1,2}=\Phi^3_{0,1}, \qquad \Phi^2_{1,2}=0, $$
\end{itemize}
where $p\leq v_2$ is the flux threshold and $\deltac$ is the complementary of the Kronecker delta as introduced in the previous Section~\ref{sect:analysis_1-2}.

\begin{figure}[!t]
\centering
\includegraphics[width=0.5\textwidth]{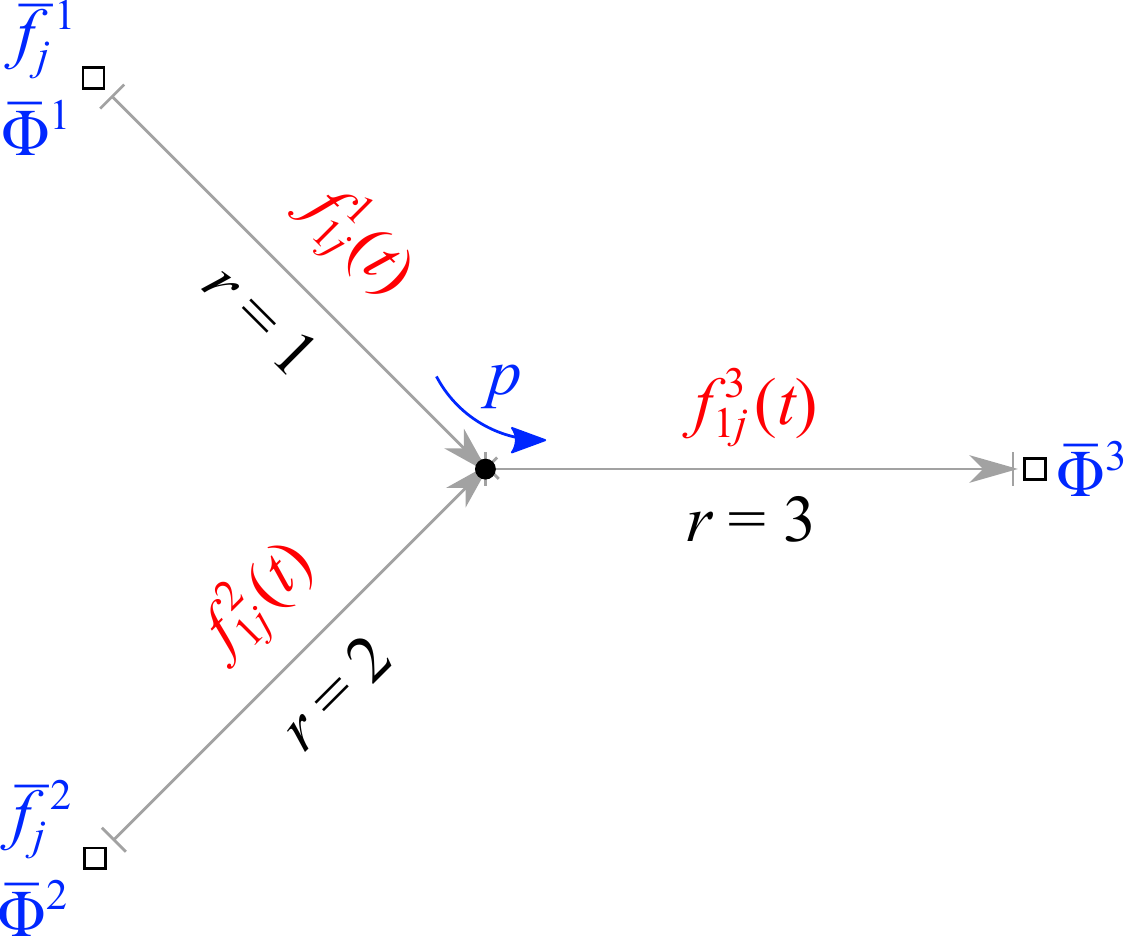}
\caption{Data (blue) and variables (red) for the problem of the 2-1 junction}
\label{fig:2-1_junction-analysis}
\end{figure}

Boundary conditions are like in Fig.~\ref{fig:2-1_junction-analysis}. In particular, at the beginning of the incoming roads we prescribe the distributions of cars and the flux limiters:
\begin{equation}
	\begin{cases}
		f^r_{0j}=\bar{f}^r_j, & j=1,\,\dots,\,n \\[1mm]
		\Phi^r_{0,1}=\bar{\Phi}^r,
	\end{cases}
	\qquad r=1,\,2.
	\label{eq:bc_2-1_1}
\end{equation}
Conversely, at the end of road $r=3$ we prescribe the flux limiter:
\begin{equation}
	\Phi^3_{1,2}=\bar{\Phi}^3.
	\label{eq:bc_2-1_2}
\end{equation}

Before tackling the qualitative analysis of this junction, a remark about the formulation of the model is in order. The transmission conditions above imply that the distribution function $f^3_{0j}$ can be discontinuous in time when passing from the case $q^1_1+q^2_1\leq p$ to the case $q^1_1+q^2_1>p$ (or vice versa). In fact, cars from road $r=2$ suddenly stop flowing into road $r=3$ as soon as the priority rule takes effect (and, analogously, they suddenly start flowing again when the application of the rule ceases). This is easily guessed to cause a lack of continuous dependence of the solution on the data. On the other hand, the chosen implementation of the right-of-way rule is not the only possible one: other smoother transitions from double to single flow across the junction can be devised, inspired by the same underlying physical ideas. In the next section we propose one such regularization of the 2-1 junction, which we will then use for our well-posedness theory.

\subsubsection{The regularized 2-1 junction}
Let $\epsilon>0$ be fixed and let $H_\epsilon:\R\to\R$ be a Lipschitz continuous function such that $H_\epsilon(x)=0$ if $x\leq 0$, $H_\epsilon(x)=1$ if $x\geq\epsilon$, and $0\leq H_\epsilon(x)\leq 1$ if $0<x<\epsilon$ (in practice, $H_\epsilon$ is a regularized version of the Heaviside function). Next let us define the following \emph{``regularized'' right-of-way rule}:
\begin{equation}
	q^3_0=q^1_1+H_\epsilon(p-(q^1_1+q^2_1))q^2_1.
	\label{eq:rightofway_reg}
\end{equation}

This definition eliminates the discontinuity intrinsic in rule~\eqref{eq:rightofway_1}-\eqref{eq:rightofway_2} thanks to the due regularity of the function $H_\epsilon$. In particular, Eq.~\eqref{eq:rightofway_reg} states that the incoming flux at the beginning of road $r=3$ switches smoothly from $q^3_0=q^1_1+q^2_1$, when $q^1_1+q^2_1\leq p-\epsilon$, to $q^3_0=q^1_1$, when $q^1_1+q^2_1\geq p$, passing through the intermediate state $q^3_0=q^1_1+H_\epsilon q^2_1$, when $p-\epsilon<q^1_1+q^2_1<p$, in which the contribution of road $r=2$ is progressively reduced. The steepness of the transition depends obviously on $\epsilon$.

Rule~\eqref{eq:rightofway_reg}, along with the usual parallelism with the 1-1 junction, implies that the distribution function at the beginning of road $r=3$ is
\begin{equation}
	f^3_{0j}=\left[f^1_{1j}+H_\epsilon(p-(q^1_1+q^2_1))f^2_{1j}\right]\deltac
	\label{eq:f0j_2-1_reg}
\end{equation}
whence, according to the mass conservation~\eqref{eq:masscons_2-1} at the junction, it also results:
\begin{equation}
	\begin{cases}
		\Phi^{1}_{1,2}=\Phi^3_{0,1} \\
		\Phi^{2}_{1,2}=H_\epsilon(p-(q^1_{1}+q^2_{1}))\Phi^3_{0,1}.
	\end{cases}
	\label{eq:Phi_2-1_reg}
\end{equation}

\begin{remark}
By letting $\epsilon\to 0^+$ the function $H_\epsilon$ formally converges to the Heaviside function. It is then immediate to see that Eqs.~\eqref{eq:rightofway_reg}--\eqref{eq:Phi_2-1_reg} reduce in that case to Eqs.~\eqref{eq:rightofway_1}--\eqref{eq:Phi_2-1_2} of the original model.
\end{remark}

\begin{remark}
Model~\eqref{eq:rightofway_reg}--\eqref{eq:Phi_2-1_reg} is, in general, not a conservative one for the free flux at the junction, in fact $q^3_0\ne q^1_1+q^2_1$ when $p-\epsilon<q^1_1+q^2_1<p$. However, as Eq.~\eqref{eq:masscons_2-1} demonstrates, continuity of the free flux is \emph{not} necessary for mass conservation purposes, the only important one being the continuity of the constrained flux characterizing the transport term of the kinetic equation~\eqref{eq:single.road.eq}.
\end{remark}

Concerning the threshold $p$, we mention that the same argument proposed in Section~\ref{sect:2-1_junction} shows that $p\leq v_2$ produces an admissible density at the beginning of road $r=3$ (that is, $\rho^3_0\leq 1$).

\subsubsection{Well-posedness of the regularized 2-1 junction}
Plugging conditions~\eqref{eq:f0j_2-1_reg},~\eqref{eq:Phi_2-1_reg} into Eq.~\eqref{eq:single.road.eq} and taking the boundary conditions~\eqref{eq:bc_2-1_1},~\eqref{eq:bc_2-1_2} into account we obtain the following problem for the 2-1 junction:
\begin{equation}
	\begin{cases}
		\dfrac{df^1_{1j}}{dt}+v_j(\bar{\Phi}^1\bar{f}^1_{j}-\Phi^3_{0,1}f^1_{1j})=G_j[\f^1_1,\,\f^1_1]-f^1_{1j}L[\f^1_1] \\[3mm]
		\begin{aligned}[t]
			\dfrac{df^2_{1j}}{dt}+v_j(\bar{\Phi}^2\bar{f}^2_{j}-H_\epsilon(p-(q^1_1+q^2_1))\Phi^3_{0,1}f^2_{1j}&) \\
				& =G_j[\f^2_1,\,\f^2_1]-f^2_{1j}L[\f^2_1] \\[3mm]
			\dfrac{df^3_{1j}}{dt}+v_j\bigl(\Phi^3_{0,1}(f^1_{1j}+H_{\epsilon}(p-(q^1_1+q^2_1))f^2_{1j}) &\deltac -\bar{\Phi}^3f^3_{1j})\bigr) \\
				& =G_j[\f^3_1,\,\f^3_1]-f^3_{1j}L[\f^3_1]
		\end{aligned} \\[3mm]
		f^r_{1j}(0)=\varphi^r_j \quad (r=1,\,2,\,3)
	\end{cases}
	\label{eq:IVP_2-1}
\end{equation}
which in mild formulation reads:
\begin{equation}
	\begin{cases}
		f^1_{1j}(t)=\varphi^1_j+\ds{\int_0^t}\left\{v_j(\bar{\Phi}^1\bar{f}^1_j(s)-\Phi^3_{0,1}f^1_{1j}(s))+G_j[\f^1_1,\,\f^1_1]-f^1_{1j}(s)L[\f^1_1]\right\}\,ds \\[3mm]
		\begin{aligned}[t]
			f^2_{1j}(t)=\varphi^2_j+\int_0^t\bigl\{&v_j\left(\bar{\Phi}^2\bar{f}^2_j-H_\epsilon(p-(q^1_1+q^2_1))\Phi^3_{0,1}f^2_{1j}(s)\right) \\
				& +G_j[\f^2_1,\,\f^2_1]-f^2_{1j}(s)L[\f^2_1]\bigr\}\,ds \\
			f^3_{1j}(t)=\varphi^3_j+\int_0^t\bigl\{&v_j\left(\Phi^3_{0,1}(f^1_{1j}(s)+H_{\epsilon}(p-(q^1_1+q^2_1))f^2_{1j}(s))\deltac-\bar{\Phi}^3f^3_{1j}(s)\right) \\
				& +G_j[\f^3_1,\,\f^3_1]-f^3_{1j}(s)L[\f^3_1]\bigr\}\,ds. \\
		\end{aligned}
	\end{cases}
	\label{eq:IVPmild_2-1}
\end{equation}

Now we are in a position to state the main results concerning the mild kinetic solution across the 2-1 junction. For the sake of definiteness, in the following we recall the definition of mild solution along with the assumptions that bring to our results.

\begin{definition}
A distribution function
$$ \F_1(t)=(\f^1_1(t),\,\f^2_1(t),\,\f^3_1(t)):[0,\,T]\to\Rn\times\Rn\times\Rn, $$
where $\f^r_1=\{f^r_{1j}\}_{j=1}^{n}$, $r=1,\,2,\,3$, is said to be a \emph{mild kinetic solution across the 2-1 junction} if $\F_1\in\B^3$ and $\F_1$ satisfies Eq.~\eqref{eq:IVPmild_2-1}.
\label{def:mild.solution_2-1}
\end{definition}

\begin{framed}
\centering\textbf{Assumptions for Problem~\eqref{eq:IVP_2-1}}
\begin{enumerate}
\item Initial and boundary data are such that $\bvarphi^r,\,\bar{\f}^1,\,\bar{\f}^2\in\B$, $r=1,\,2,\,3$.
\item The function $\Phi:[0,\,1]^2\to [0,\,1]$ defining the flux limiter $\Phi^3_{0,1}=\Phi(\rho^3_0,\,\rho^3_1)$ is like in Assumption~\ref{ip2} of Problem~\eqref{eq:IVP_1-2}, cf. page~\pageref{ip2}.
\item The flux limiters $\bar{\Phi}^1$, $\bar{\Phi}^2$ (boundary data) at the entrance of the incoming roads are like in Assumption~\ref{ip2} of Problem~\eqref{eq:IVP_1-2}, cf. page~\pageref{ip2}. In particular, for $r=1,\,2$, we have $\bar{\Phi}^r=\Phi(\bar{\rho}^r,\,\rho^r_1)$, where $\bar{\rho}^r=\sum_{j=1}^{n}\bar{f}^r_{j}$ is the density of cars accessing road $r$ according to the boundary datum $\bar{\f}^r$.
\item The elements $A^j_{hk}$ of the table of games are like in Assumption~\ref{ip4} of Problem~\eqref{eq:IVP_1-2}, cf. page~\pageref{ip4}.
\end{enumerate}	
\end{framed}

Similarly to the case of the 1-2 junction, we state three results. The first one gives a continuous dependence estimate for the mild solution of Problem~\eqref{eq:IVP_2-1}, whence in particular its uniqueness; the second one gives its existence; finally, the third one shows that the mild solution thus found is actually a classical (i.e., differentiable in time) one. We omit to detail the proofs since they are obtained by means of the very same procedures employed in the case of the 1-2 junction, up to taking into account now the suitable smoothness properties of the function $H_\epsilon$.

\begin{theorem}[Uniqueness and continuous dependence for the (regularized) 2-1 junction]
Let $\{\{\bvarphi^r\}_{r=1}^{3},\,\{\bar{\f}^r\}_{r=1}^{2},\,\{\bar{\Phi}^r_f\}_{r=1}^{3}\}$, $\{\{\bphi^r\}_{r=1}^{3},\,\{\bar{\g}^r\}_{r=1}^{2},\,\{\bar{\Phi}^r_g\}_{r=1}^3\}$ be two sets of initial and boundary data for the 2-1 junction and $\F_1,\,\G_1\in\B^3$ two corresponding mild solutions of Problem~\eqref{eq:IVP_2-1}. Then there exists $\C>0$ such that
\begin{multline*}
	\norm{\G_1-\F_1}_\infty\leq\C\left[\sum_{r=1}^{3}\norm{\bphi^r-\bvarphi^r}_1\right. \\
		\left.+\int_0^T\left(\sum_{r=1}^{3}\abs{\bar{\Phi}^r_g(t)-\bar{\Phi}^r_f(t)}+\sum_{r=1}^{2}\norm{\bar{\g}^r(t)-\bar{\f}^r(t)}_1\right)dt\right].
\end{multline*}
In particular, there is at most one solution corresponding to a given set of initial and boundary data.
\end{theorem}

\begin{theorem}[Existence for the (regularized) 2-1 junction]
There exists $\F_1\in\B^3$ which is a mild solution to Problem~\eqref{eq:IVP_2-1} in the sense of Definition~\ref{def:mild.solution_2-1}.
\end{theorem}

\begin{corollary}[Improved regularity for the solution of the (regularized) 2-1 junction]
The mild solution to Problem~\eqref{eq:IVP_2-1} is of class $C^1$ in $[0,\,T]$.
\end{corollary}

\bibliographystyle{plain}
\bibliography{FlTa-networks}
\end{document}